\documentclass[11pt]{article}

\usepackage[pdftex]{graphicx}
\usepackage{fullpage,amsmath,ams fonts, amssymb,amsthm}
\usepackage{multirow}

\newtheorem{theorem}{Theorem}[section]

\newtheorem{lemma}[theorem]{Lemma}
\newtheorem{corollary}[theorem]{Corollary}

\newtheorem{proposition}[theorem]{Proposition}
\newtheorem{definition}[theorem]{Definition}

\newtheorem{problem}[theorem]{Problem}
\newtheorem{openproblem}[theorem]{Open Problem}

\newcommand{\algwidth}{0.97\textwidth}

\newcommand{\cA}{\mathcal{A}}
\newcommand{\cB}{\mathcal{B}}

\newcommand{\cE}{\mathcal{E}}

\newcommand{\cX}{\mathcal{X}}
\newcommand{\cY}{\mathcal{Y}}

\newcommand{\cU}{\mathcal{U}}
\newcommand{\cT}{\mathcal{T}}

\newcommand{\cW}{\mathcal{W}}

\newcommand{\E}{{\mathbb{E}}}
\newcommand{\eps}{\varepsilon}

\newcommand{\X}{\mathbf{X}}
\newcommand{\Y}{\mathbf{Y}}

\usepackage{xspace}

\newcommand{\poly}{{\operatorname{poly}\xspace}}

\newcommand{\Dmu}[3]{{\mathsf{D}_{#2}(#1,#3)}}
\newcommand{\Rmu}[3]{{\mathsf{R}(#1,#2)}}
\newcommand{\ICmu}[3]{{\mathsf{IC}_{#2}(#1,#3)}}

\newcommand{\ICprot}[2]{{\mathsf{IC}_{#2}(#1)}}
\newcommand{\ICprotE}[2]{{\mathsf{IC}^{\mathsf{ext}}_{#2}(#1)}}

\newcommand{\Div}[2]{\mathbb{D} \left (  #1 \| #2 \right) }

\newcommand{\suc}{\mathsf{suc}}
\newcommand{\suci}{\mathsf{suc^i}}

\makeatletter
\newcommand\footnoteref[1]{\protected@xdef\@thefnmark{\ref{#1}}\@footnotemark}
\makeatother

\begin{document}

\title{Information Complexity and the Quest for Interactive Compression \\ \vspace{.13in} \small (Survey)}

\author{Omri Weinstein\thanks{Department of Computer Science, Princeton University, 
%Princeton, NJ 08540, USA\@ 
{\tt oweinste@cs.princeton.edu}. 
Research supported by a Simons award in Theoretical Computer Science, a Siebel scholarship and  NSF Award CCF-1215990.
A version of this survey will appear in the June 2015 issue of SIGACT News complexity column.}}

\sloppy
\date{}
\maketitle

\pagestyle{empty}

\begin{abstract}
Information complexity is the interactive analogue of Shannon's classical information theory. In recent years 
this field has emerged as a powerful tool for proving strong communication lower bounds, and for addressing 
some of the major open problems  in  communication complexity and circuit complexity. 
A notable achievement of information complexity is the 
breakthrough in understanding of  the fundamental direct sum and direct product conjectures,
which aim to quantify the power of parallel computation. 
This survey provides a brief introduction to information complexity, and overviews some of the
recent progress on these conjectures and their tight relationship with the fascinating problem of 
compressing interactive protocols.
\end{abstract}

\section{Introduction}
The holy grail of complexity theory is proving lower bounds on different computational models, thereby delimiting 
computational problems according to the resources they require for solving. 
One of the most useful abstractions for proving such lower bounds is \emph{communication complexity}; 
Since its introduction \cite{Yao79}, this model has had a profound impact on nearly every field of theoretical computer science,
including VLSI chip design, data structures, mechanism design, property testing and streaming algorithms \cite{Thompson79,PW10,DN11,BBM12} to mention a few,
and constitutes one of the few known tools for proving \emph{unconditional} lower bounds. 
As such, developing new tools in communication complexity is a promising approaches for making progress within computational complexity, 
and in particular, for proving strong circuit lower bounds that appear viable (such as Karchmer-Wigderson games and ACC lower bounds \cite{KW90, BeigelT91}). 

Of particular interest are ``black box" techniques for proving lower bounds, also known as ``hardness amplification" methods (which morally enable strong
lower bounds on composite problems via lower bounds on a simpler primitive problem). 
Classical examples of such results are the Parallel Repetition theorem \cite{Raz98, Rao08a} and Yao's XOR Lemma  \cite{Yao82}, both of which are cornerstones 
of complexity theory. This is the principal motivation for studying the \emph{direct sum and direct product} conjectures, which are at the core of this survey.

Perhaps the most notable tool for studying communication problems is information theory, introduced by Shannon in the late 1940Õs in the context of 
(one-way) data transmission problems \cite{Shannon48}. 
Shannon's noiseless coding theorem revealed the tight connection between communication and information, namely, that the amortized description length 
of a random one-way message ($M$) is equivalent to the amount of information it contains
\begin{align} \label{eq_shannon_noiseless}
\lim_{n\longrightarrow \infty} \frac{C(M^n)}{n} = H(M),
\end{align}
where $M^n$ denotes $n$ i.i.d observations from $M$, $C$ is the minimum number of bits of a string from which $M^n$ can be recovered (w.h.p),
and $H(\cdot)$ is Shannon's Entropy function.
In the 65 years that elapsed since then, information theory has 
been widely applied and developed, and has become the primary mathematical tool for analyzing communication problems.

%\paragraph{Relation and differentiation from classical information theory.}
%Classical information theory studies the setting where one terminal (Alice) wishes to transmit some message $X$ over a noiseless channel to another terminal (Bob). 
Although classical information theory provides a complete understanding of the one-way transmission setup (where only one party communicates), it does not readily convert to the 
\emph{interactive setup}, such as the (two-party) communication complexity model. 
In this model, two players (Alice and Bob) receive inputs $x$ and $y$ respectively, which are jointly distributed according to some prior distribution $\mu$,
and wish to compute some function $f(x,y)$ while communicating as little as possible. To do so, they engage in a \emph{communication protocol}, and are allowed to
use both public and private randomness.  
A natural extension of Shannon's entropy to the interactive setting is the \emph{Information Complexity} of a function $\ICmu{f}{\mu}{\eps}$, which informally measures the 
average amount of information the players need to disclose 
each other about their inputs in order to solve $f$ with some prescribed error under the input distribution $\mu$. %(see the for following section).
From this perspective, communication complexity can be viewed as the extension of transmission problems to general tasks performed by two parties over a 
noiseless channel (the noisy case recently received a lot of attention as well \cite{BraICM14}). 
Interestingly, it turns out that an analogue of Shannon's theorem does in fact hold for interactive computation,
asserting that the amortized communication cost of computing many independent copies of any function $f$ is precisely equal to its single-copy information complexity:
\begin{theorem}[``Information $=$ Amortized Communication", \cite{BravermanR11}]  \label{thm_ic_eq_am_cc}
For any $\eps>0$ and any two-party communication function $f(x,y)$, %it holds that
\[ \lim_{n\longrightarrow \infty}\frac{\Dmu{f^n}{\mu^n}{\eps}}{n} = \ICmu{f}{\mu}{\eps}. \]
\end{theorem}
\noindent Here $\Dmu{f^n}{\mu^n}{\eps}$ denotes the minimum communication required for solving $n$ independent instances of $f$ with error at most $\eps$ on 
\it each copy\rm.\footnote{Indeed,
the ``$\le$" direction of this proof gives a protocol with overall success only $\approx (1-\eps)^n$ on all $n$ copies, %and this is essentially tight, 
see \cite{BravermanR11}.}
The above theorem assigns an \emph{operational} meaning to information complexity, namely, one which is grounded in reality
(in fact, it was recently shown that this characterization is a ``sharp threshold", see Theorem \ref{thm:dp_for_ic}). % in Section \ref{subsec_ds_dp}).

Theorem \ref{thm_ic_eq_am_cc} and some of the additional results mentioned in this survey, provide a strong evidence that 
information theory is the ``right" tool for studying interactive communication problems.  
One general benefits of information theory in addressing communication problems is that
it provides a set of simple yet powerful tools for reasoning about transmission problems and more broadly about quantifying relationships between 
interdependent random variables and conditional events. Tools that include mutual information, 
the chain rule, and the data processing inequality \cite{CoverT91}. 
%Indeed, a remarkable feature of information complexity, which stems from these simple tools, is that it is a fully additive measure over 
%composition\footnote{$T1 \otimes T2$ denotes the task composed of 
%successfully performing both $T_1$ and $T_2$ on the respective inputs $(X1,Y1)\sim\mu_1$ and $(X2,Y2)\sim\mu_2$.} of tasks:
%\begin{align} \label{lem_additivity_of_ic}
% \IC_{\mu_1 \times \mu_2}(T_1 \otimes T_2) = \IC_{\mu_1}(T_1) + \IC_{\mu_2}(T_2).
%\end{align}
Another, arguably most important benefit, is the \emph{additivity} of information complexity under composition of independent tasks 
(Lemma \ref{lem_additivity_of_ic} below). This is the main reason that information theory, unlike other analytic or combinatorial methods,  
is apt to give \emph{exact} bounds on rates and capacities (such as Shannon's noiseless coding theorem and Theorem \ref{thm_ic_eq_am_cc}). 
It is this benefit that has been primarily used in prior works (which are beyond the scope of this survey)
involving information-theoretic applications in communication complexity, circuit complexity, streaming, machine learning and privacy (\cite{%ablayev1996lower,
ChakrabartiSWY01, LS10, CKS03,BaryossefJKS04,JayramKR09,BGPW13,ZDJW13,WZ14} to mention a few). 

%In fact, unlike computational complexity, where we often ignore linear, polylogarithmic, and sometimes even polynomial factors, 
%a large fraction of results in information theory provide us with precise answers up to additive lower-order terms. 
%For instance, we know that a sequence of random digits would take exactly $\log_210\approx 3.322$ bits per digit, and that the capacity 
%of a binary symmetric channel with substitution probability $0.2$ is exactly $1-H(0.2)\approx 0.278$ bits per symbol. 
%A program which has emerged in the field over the past few years is to understand whether such fundamental results translate into the interactive setup.
%While this program is only at its preliminary stage, we provide encouraging results in this direction -- For example, 
%the tools we developed enabled us to determine the \emph{exact} communication complexity  of the Set-Disjointness function, which turns out to be surprisingly 
%low: $C_{DISJ}\cdot n \approx 0.48n$   (see Section \ref{subsec_from_ic_to_exact_cc}). Such precise results were beyond the reach of analytical techniques before the 
%emergence of information complexity. \\

One caveat is that mathematically striking characterizations such as the noiseless coding theorem only become possible in the limit, where the number of
independent samples transmitted over the channel (i.e., the block-length) grows to infinity. One exception is Huffman's \emph{``one-shot" } compression scheme  (aka Huffman coding, \cite{huffman1952method}), 
which shows
that the expected number of bits $C(M)$ needed to transmit a \emph{single sample} from $M$, is very close (but not equal!) to the optimal rate
\begin{align} \label{eq_huffman}
H(M) \leq C(M) \leq H(M)+1.
\end{align}
Huffman's theorem of course implies Shannon's theorem (since entropy is additive over independent variables),  but is in fact
much stronger, as it asserts that the optimal transmission rate can be (essentially) achieved using much a smaller block length. 
Indeed, what happens for small block lengths 
is of importance for both practical and theoretical reasons, and it will be even more so in the interactive regime.  
While Theorem \ref{thm_ic_eq_am_cc} provides an interactive analogue of  Shannon's theorem, 
an intriguing question is whether an interactive analogue  of Huffman's ``one-shot" compression scheme exists.
When the number of communication rounds of a protocol is small (constant),
compressing it can morally\footnote{This is not 
accurate, since unlike the one-way transmission setting, in this setting the receiver has ``side information" about the transmitted message, e.g., when 
Bob sends the second message of the protocol, %$(M_2|Y,M_1)$, 
Alice has a prior distribution on this message %$(M_2|X,M_1)$}
conditioned on her input $X$ and the first message of the protocol $M_1$ which she sent before. Nevertheless, using ideas from rejection sampling, 
such simulation is possible in the ``one-shot" regime with $O(1)$ communication overhead per message \cite{HJMR07, BravermanR11}.}
be done by applying Huffman's compression scheme to each round of the protocol, since \eqref{eq_huffman} would entail at most a constant
overhead in communication. However, when the number of rounds is huge compared to the overall information revealed by the protocol (e.g., when each round 
reveals $\ll 1$ bits of 
information), this approach is doomed to fail, as it would ``spend" at least $1$ bit of communication per round.
Circumventing this major obstacle and the important implications of this (unsettled) question 
to the direct sum and product conjectures %in communication complexity 
are extensively discussed  in Sections \ref{subsec_ds_dp} and \ref{sec_compression}.

%Indeed, the above distinction between amortized and so called ``one-shot" results is one of the main distinguishing features of
%information complexity from classic information theory.
%Another distinctive aspect is in that communication complexity often studies functions whose output is only a single bit or a small number of bits, thus \mnote{rephrase par}
%Òcounting styleÓ direct lower bound proofs rarely apply (a notable example is the SetDisjointness function mentioned above $F(X,Y)=\bigwedge_i (\neg X_i \vee \neg Y_i)$, whose $\Omega(n)$ 
%linear lower bound challenged the TCS community for almost two decades -- intuitively, the lower bound must argue that despite the succinct evidence for intersection ($1$ bit), the parties must ``waste" a lot of communication on finding it). \\

%%%%%%%%%%%%%%%%%%%%%%%%%%%%%%%%%%%%%%%%%%%%%%%%%%%%%%%%%%%%%%%%%%%%%%%%%%%%%%%%%%%%%%%%
%%%%%%%%%%%%%%%%%%%%%%%%%%%%%%%%%%%%%%%%%%%%%%%%%%%%%%%%%%%%%%%%%%%%%%%%%%%%%%%%%%%%%%%%
%%%%%%%%%%%%%%%%%%%%%%%%%%%%%%%%%%%%%%%%%%%%%%%%%%%%%%%%%%%%%%%%%%%%%%%%%%%%%%%%%%%%%%%%
%%%%%%%%%%%%%%%%%%%%%%%%%%%%%%%%%%%%%%%%%%%%%%%%%%%%%%%%%%%%%%%%%%%%%%%%%%%%%%%%%%%%%%%%

%\section{Organization}
Due to space constraints, this survey is primarily focused on the above relationship between information and communication complexity. 
As mentioned above, information complexity has recently found many more exciting applications in complexity theory -- 
to interactive coding, streaming lower bounds, extension complexity and  multiparty communication complexity (e.g., \cite{BaryossefJKS04, BM12, BP13, BOEPV13}).
Such applications are beyond the scope of this survey.
%For a broader introduction to information complexity and applications of this powerful tool, we refer the reader  to excellent monographs by 
%Braverman [ICM, STOC paper] and [applications?].

%%%%%%%%%%%%%%%%%%%%%%%%%%%%%%%%%%%%%%%%%%%%%%%%%%%%%%%%%%%%%%%%%%%%%%%%%%%%%%%%%%%%%%%%
%%%%%%%%%%%%%%%%%%%%%%%%%%%%%%%%%%%%%%%%%%%%%%%%%%%%%%%%%%%%%%%%%%%%%%%%%%%%%%%%%%%%%%%%
%%%%%%%%%%%%%%%%%%%%%%%%%%%%%%%%%%%%%%%%%%%%%%%%%%%%%%%%%%%%%%%%%%%%%%%%%%%%%%%%%%%%%%%%
%%%%%%%%%%%%%%%%%%%%%%%%%%%%%%%%%%%%%%%%%%%%%%%%%%%%%%%%%%%%%%%%%%%%%%%%%%%%%%%%%%%%%%%%

\paragraph{Organization}
We begin with a brief introduction to information complexity and some of its main properties (Section \ref{sec_prelims}). In Section
\ref{subsec_ds_dp} we give an overview of the direct sum and direct product conjectures and their relationship 
to interactive compression, in light of recent developments in the field. 
Section \ref{sec_compression} describes state-of-the-art interactive compression schemes. We conclude with 
several natural open problems in Section \ref{sec_openprobs}. 
%This survey is written in an expository fashion, and...
In an effort to keep this survey as readable and self-contained as possible, we shall sometimes be 
loose on technical formulations, often ignoring constant and technical details which are not essential to the reader. 

\section{Model and Preliminaries} \label{sec_prelims}

The following background contains basic definitions and notations used throughout this survey. For a more detailed overview
of communication and information complexity, we refer the reader to an excellent monograph by Braverman \cite{Bra12}.

%\begin{definition}[Distributional and Randomized Communication Complexity]
For a function $f: \mathcal{X} \times \mathcal{Y} \rightarrow \mathcal{Z}$, a distribution $\mu$ over$\cX\times\cY$, 
%supported on $\mathcal{X}\times \mathcal{Y}$, 
and a parameter $\eps > 0$,  $\Dmu{f}{\mu}{\eps}$ denotes the communication complexity of the cheapest
deterministic protocol computing $f$ on inputs sampled according to $\mu$ with error $\eps$. 
%When clear from context, we use the notation $\mathsf{D}_\mu(f)$ for the \emph{deterministic} communication complexity of $f$.
$\Rmu{f}{\eps}{}$ denotes the cost of the best \emph{randomized} public coin protocol which computes $f$ with error at most $\eps$,
for \emph{all} possible inputs $(x,y)\in \cX\times\cY$. When measuring the communication cost of a particular protocol $\pi$, we sometimes use
the notation $\|\pi\|$ for brevity. 
%\end{definition}
Essentially all results in this survey are proven in the former \emph{distributional} communication
model (since information complexity is meaningless without a prior distribution on inputs), but most lower bounds below 
can be extended to the randomized model via Yao's minimax theorem.
For the sake of concreteness, all of the results in this article are stated for (total) functions, though most of them apply to partial 
functions and relations as well.

\subsection{Information Theory}

%The following background in information theory will be useful throughout this survey. For a broader introduction to the field, and 
Proofs of the claims below and a broader introduction to information theory can be found in \cite{CoverT91}.
The most basic concept in information theory is Shannon's entropy, which informally captures how predictable a random variable is:
\begin{definition}[Entropy]
 The \emph{entropy} of a random variable $A$ is $H(A) := \sum_a \Pr[A=a] \log(1/\Pr[A=a]).$
The \emph{conditional entropy} $H(A|B)$ is defined as $\E_{b \sim B} \left[ H(A|B=b) \right]$.
\end{definition}

A key measure in this article is the \emph{Mutual Information} between two random variables, %which is a measure of correlation between two random variables:
which quantifies the amount of  correlation between them:
\begin{definition}[Mutual Information]
The \emph{mutual information} between two random variables $A,B$, denoted $I(A;B)$ is defined to be the quantity 
$H(A) - H(A|B) = H(B) - H(B|A).$ The \emph{conditional mutual information} $I(A;B |C)$ is $H(A|C) - H(A|BC)$.
\end{definition}

A related distance measure between \emph{distributions} is the \it Kullback-Leibler \rm (KL) divergence 
\[ \Div{p}{q} := \sum_{x } p(x) \log \frac{p(x)}{q(x)} = \E_{x\sim p} \left[ \log  \frac{p(x)}{q(x)} \right].\]

We shall sometimes abuse the notation and write $\Div{A|c}{B|c}$ to denote the KL divergence between the associated
distributions of the random variables $(A|C=c)$ and $(B|C=c)$. The following connection between divergence and mutual information 
is well known:
\begin{lemma}[Mutual information in terms of Divergence] \label{prop_dim_IC} 
\[ I(A;B | C) = \E_{b,c} \left[ \Div{A|bc}{A|c} \right] = \E_{a,c} \left[ \Div{B|ac}{B|c} \right]. \]
\end{lemma}
Intuitively, the above equation asserts that, if the mutual information between $A$ and $B$ (conditioned on $C$) is large, then the distribution 
of $(A|c)$ is ``far" from $(A|bc)$ for average values of $b,c$ (this captures the fact that the ``additional information" $B$ provides on $A$ given $C$ is large).
One of the most useful properties of Mutual Information and KL Divergence is the chain rule:
\begin{lemma}[Chain Rule] 
Let $A, B, C, D$ be four random variables in the same probability space. Then
\[ I(AB;C|D) = I(A;C|D) + I(B;C|AD)\]
\[ = \E_{c,d}\left[ \Div{A|cd}{A|d}\right] + \E_{a,c,d}\left[ \Div{B|acd}{A|ad}\right] . \]
\end{lemma}

\begin{lemma}[Conditioning on independent variables does not decrease information] \label{lem_cond_indep}
Let $A,B,C,D$ be four random variables in the same probability space. If $A$ and $D$ are conditionally independent given $C$, 
then it holds that %conditioning on $D$ can only increase the mutual information between 
$I(A;B|C) \leq I(A;B|CD)$.
\end{lemma}

\begin{proof}
We apply the chain rule for mutual information twice. On one hand, we have
$ I(A;BD|C) = I(A;B|C) + I(A;D|CB) \geq I(A;B|C)$
since mutual information is nonnegative. On the other hand,
$ I(A;BD|C) = I(A;D|C) + I(A;B|CD) = I(A;B|CD)$
since $I(A;D|C)=0$ by the independence assumption on $A$ and $D$. Combining both equations completes the proof.
\end{proof}

Throughout this article, we denote by $|p-q|$ the \it total variation \rm distance between the distributions $p$ and $q$. 
Pinsker's inequality bounds statistical distance in terms of the KL divergence.
It will be useful for analysis of the interactive compression schemes in Section \ref{sec_compression}.

\begin{lemma}[Pinsker's inequality] \label{lemma:pinsker} 
%If $p(b)= q(b)$, then 
%$|p - q|^2 \leq \Ex{p}{\Div{p}{q}}$. 
$|p - q|^2 \leq \frac{1}{2}\cdot \Div{p}{q}$. 
\end{lemma}

\subsection{Interactive Information complexity}

Given a communication protocol $\pi$, $\pi(x,y)$ denotes the concatenation of the public randomness with all the messages that are
sent during the execution of $\pi$ (for information purposes, this is without loss of generality, since the public string $R$ conveys no information about the inputs). 
We call this the \emph{transcript} of the protocol. When referring to the random variable 
denoting the transcript, rather than a specific transcript, we will use the notation $\Pi(x,y)$ --- or simply $\Pi$ when $x$ and $y$ 
are clear from the context.%, thus $\pi(x,y)\getsr \Pi(x,y)$. 
%When $x$ and $y$ are 
%random variables themselves, we will denote the transcript by $\Pi(X,Y)$, or just $\Pi$. 
\begin{definition}[Internal Information Cost  \cite{ChakrabartiSWY01, BBCR}]
The \it (internal) information cost \rm of a protocol over inputs drawn from a distribution 
$\mu$ on $\cX\times\cY$, is given by: 
\begin{align}      \label{def:infoProt}
\ICprot{\pi}{\mu}:=I(\Pi;X|Y)+I(\Pi;Y|X).
\end{align}
\end{definition}
Intuitively, the definition in \eqref{def:infoProt} captures how much additional information the two parties learn about each other's inputs by observing the
protocol's  transcript.  For example, the information cost of the trivial protocol in which Alice and Bob simply exchange their inputs, is simply the sum of their
conditional marginal entropies $H(X|Y) + H(Y|X)$ (notice that, in contrast, the \emph{communication} cost of this protocol is $|X|+|Y|$ which can be arbitrarily 
larger than the former quantity).

%The first term captures what the second player learns about $X$ from $\Pi$ -- the mutual information between the input $X$ and the transcript 
%$\Pi$ given the input $Y$. 
Another information measure which makes sense at certain contexts is the \emph{external} information cost of a protocol, $\ICprotE{\pi}{\mu}:=I(\Pi;XY),$
which captures what an \emph{external} observer learns on average  about both player's inputs by observing the transcript of $\pi$. This quantity will be 
of minor interest in this survey (though it playes a central role in many applications). 
The external information cost of a protocol is always at least as large as its (internal) information cost, since intuitively an external observer is ``more
ignorant" to begin with. We remark that when $\mu$ is a \emph{product} distribution, then $\ICprotE{\pi}{\mu}=\ICprot{\pi}{\mu}$ (see, e.g., \cite{Bra12}).
%It is also noteworthy that both definitions above are implicitly  conditioned on the public random string of the protocol $R$. However,
%we may assume without loss of generality that $R$ is ``rolled" in to the protocol transcript $\Pi$, and since $I(R;X|Y)=0$ (and vice versa), these
%definitions are equivalent. This is not the case with \emph{private randomness}, as explained in Subsection \ref{sec_priv_vs_pub_coins}.
 
%Note that the information cost of a protocol $\pi$ depends on the prior distribution $\mu$, as the mutual information between 
%the transcript $\Pi$ and the inputs depends on the prior distribution 
%on the inputs. To give an extreme example, if $\mu$ is a singleton distribution, i.e. one with $\mu(\{(x,y)\}) = 1$ for some 
%$(x,y)\in \cX\times \cY$, then $\ICprot{\pi}{\mu}=0$ for all possible 
%$\pi$, as no protocol can reveal anything to the players about the inputs that they do not already know {\em a-priori}. 
%Similarly, $\ICprot{\pi}{\mu}=0$ if $\cX=\cY$ and $\mu$ is supported on 
%the diagonal $\{(x,x):x\in\cX\}$.
%\smallskip

One can now define the \emph{information complexity} of a function $f$ with respect to $\mu$ and error $\eps$ as the least amount of 
information the players need to reveal to each other in order to compute $f$ with error at most $\eps$:

\begin{definition}
The \it Information Complexity \rm of $f$ with respect to $\mu$ (and error $\eps$) is   
$$\ICmu{f}{\mu}{\eps}:= \inf_{\pi: \; \Pr_\mu[\pi(x,y)\neq f(x,y)]\le \eps} \ICprot{\pi}{\mu}. $$
\end{definition}

What is the relationship between the information and communication complexity of $f$? This question is at the core of this survey. 
The answer to one direction is easy:
Since one bit of communication can never reveal more than one bit of information,  the communication cost of any protocol is
always an upper bound on its information cost over {\em any} distribution
$\mu$:

\begin{lemma}[\cite{BravermanR11}] \label{lem:ICCC}
 For any distribution $\mu$, $\ICprot{\pi}{\mu}\le \|\pi\|$. 
\end{lemma}

The answer to the other direction, namely, whether any protocol can be compressed to roughly its information cost, will be partially
given in the remainder of this article. \\

%The (prior-free) information cost was defined in \cite{BravermanInteractive11} as the minimum amount of information that 
%a worst-case error-$\rho$ randomized protocol can reveal against {\em all} possible prior distributions. 
%$$
%\IC{f}{\rho} := \inf_{\text{$\pi$ is a protocol with $\P[\pi(x,y)\neq f(x,y)]\le\rho$ for all $(x,y)$}} \max_\mu~~ \ICprot{\pi}{\mu}.
%$$
%This information cost matches the amortized {\em randomized} communication complexity of $f$ \cite{BravermanInteractive11}. It is 
%clear that lower bounds on $\ICmu{f}{\mu}{\rho}$ {\em for any distribution $\mu$} also apply to $\IC{f}{\rho} $.

%%%%%%%%%%%%%%%%%%%%%%%%%%%%%%%%%%%%%%%%%%%%%%%%%%%%%%%%%%%
%%%%%%%%%%%%%%%%%%%%%%%    Subsection         %%%%%%%%%%%%%%%%%%%%%%%%%%%
%%%%%%%%%%%%%%%%%%%%%%%%%%%%%%%%%%%%%%%%%%%%%%%%%%%%%%%%%%%

\subsection{The role of private randomness in information complexity} %The importance of private randomness} 
\label{sec_priv_vs_pub_coins}

A subtle but vital issue when dealing with information complexity, is understanding the role of private vs.\ public randomness.
In public-coin communication complexity, one often ignores the usage of private coins in a protocol, as they 
can always be simulated by public coins. 
%In fact,  in the distributional model, when inputs 
%arrive from some prior distribution $\mu$, even public coins are essentially redundant since the averaging principle asserts that the communication 
%cost can be minimized by some fixed choice of the public 
%randomness.
When dealing with \emph{information complexity}, the situation is somewhat the opposite: Public coins are essentially a redundant resource (as it can be easily shown via 
the chain rule that $\ICprot{\pi}{\mu} = \E_R[\ICprot{\pi_R}{\mu}]$), 
while the usage of private coins is crucial for minimizing the information cost, and fixing these coins
is prohibitive (once again, for communication purposes in the distributional model,  one may always fix the entire randomness of the protocol, 
via the averaging principle). To illustrate this point, consider the simple example where in the protocol $\pi$, Alice sends Bob her $1$-bit 
input $X\sim Ber(1/2)$, XORed with some random
bit $Z$. If $Z$ is private, Alice's message clearly reveals $0$ bits of information to Bob about $X$. However, for any fixing of $Z$, this message would reveal an entire bit(!).
The general intuition is that a protocol with low information cost would reveal information about the player's inputs in a ``careful manner", and the usage of private coins 
serves to  ``conceal" parts of their inputs. Indeed, it was recently shown that the restriction to public coins may cause an exponential blowup in the information revealed
compared to private-coin protocols (\cite{GKR14,BMY14}). In fact, we shall see in Section \ref{subsec_ds_dp} that quantifying this gap between public-coin and private-coin information complexity is tightly related to the question of interactive compression. %,which we discuss extensively below.

%This is not a theorem (and in fact, quantifying the loss in information from restricting protocols to public coins only, 
%would essentially resolve the interactive compression
%problem ,\ref{}), but many known examples (e.g., \cite{BGPW13}) support this intuition. 

For the remainder of this article, communication protocols $\pi$ are therefore assumed to use private coins (and therefore such protocols are randomized even conditioned
on the inputs $x,y$ and $R$), and it is crucial that the information cost $\ICprot{\pi}{\mu} = I(\Pi;X|YR)+I(\Pi;Y|XR)$ is measured conditioned on the \it public \rm randomness $R$, but never on the private coins of $\pi$.

\section{Additivity of Information Complexity} \label{sec_ic_additive}

Perhaps the single most remarkable property of information complexity is that it is a fully additive measure over composition of tasks. This property 
is what primarily makes information complexity a natural ``relaxation" for addressing direct sum and product theorems.
The main ingredient of the following lemma appeared first in the works of \cite{Razborov08,Raz98} and more explicitly in \cite{BBCR,BravermanR11,Bra12}.
In the following, $f^n$ denotes the function that maps the tuple $((x_1,\ldots,x_n),(y_1,\ldots,y_n))$ to  $(f(x_1,y_1),\ldots,f(x_n,y_n)))$.

\begin{lemma}[Additivity of Information Complexity] \label{lem_additivity_of_ic}
$\ICmu{f^n}{\mu^n}{\eps} = n\cdot \ICmu{f}{\mu}{\eps}$.
\end{lemma}

\begin{proof}
The ($\leq$) direction of the lemma is easy, and follows from a simple argument that applies the single-copy optimal protocol independently 
to each copy of $f^n$, with independent randomness. We leave the simple analysis of this protocol as an exercise to the reader.

The ($\geq$) direction is the main challenge. Will will prove it in a contra-positive fashion: Let $\Pi$ be an $\eps$-error protocol for $f^n$,
such that $\ICprot{\Pi}{\mu^n} = I$ (recall that here $\eps$ denotes the per-copy error of $\Pi$ in computing $f(x_i,y_i)$). 
We shall use $\Pi$ to produce a \it single-copy \rm protocol for $f$ whose information cost is $\leq I/n$,
which would complete the proof. The guiding intuition for this is that $\Pi$ should reveal $I/n$ bits of information about an average coordinate.  

To formalize this intuition, let $(x,y)\sim \mu$, and denote $\X := X_1\ldots X_n$ , $X_{\leq i} := X_1 \ldots X_{i}$ and $X_{- i} := X_1 \ldots X_{i-1},X_{i+1},\ldots, X_n$, 
and similarly for $\Y, Y_{\leq i}, Y_{-i}$.
A natural idea is for Alice and Bob to ``embed" their respective inputs $(x,y)$ to a (publicly chosen) random coordinate $i \in [n]$ of $\Pi$, and execute $\Pi$. 
However, $\Pi$ is defined over $n$ input copies, so in order to execute it, the players need to somehow ``fill in" the rest $(n-1)$ coordinates, each according to $\mu$.
How should this step be done? The first attempt is for Alice and Bob to try and complete $X_{-i},Y_{-i}$ privately. This approach fails if $\mu$
is a non-product distribution, since there's no way the players can sample $X$ and $Y$ privately, such that $(X,Y)\sim \mu$ if $\mu$ correlates
the inputs. The other extreme -- sampling $X_{-i},Y_{-i}$ using public randomness only -- would resolve the aforementioned correctness issue,
but might leak too much information: An instructive example to consider is where, in the first message of $\Pi$, Alice sends Bob the XOR of the $n$ bits of her uniform 
input $X$: $M = X_1\oplus X_2\oplus \ldots \oplus X_n$. Conditioned on $X_{-i},Y_{-i}$, $M$ reveals $1$ bit of information about $X_i$ to Bob, while we want to argue that 
in this case, only $1/n$ bits are revealed about $X_i$. So this %sampling 
approach reveals too much information.

It turns out that the ``right" way of breaking the dependence across the coordinates is to use a combination of public and private randomness. 
Let us define, for each $i\in [n]$, the public random variable %Alice and Bob will start by publicly sampling the random variable 
\[ R_i:= X_{< i},Y_{> i}. \]
Note that given $R_i$, Alice can complete all her missing inputs $X_{>i}$ \it privately \rm according to $\mu$, and Bob can do the same for $Y_{<i}$.
Let us denote by $\theta(x,y,i,R_i)$ the %distribution on 
protocol transcript produced by running $\Pi(X_1,...,X_{i-1},x,X_{i+1},...,X_n \;, \; Y_1,...,Y_{i-1},y,Y_{i+1},...,Y_n)$ and outputting its answer on the $i$'th coordinate. 
Let $\Theta(x,y)$ be the protocol obtained by running $\theta(x,y,i,R_i)$ on a uniformly selected $i\in [n]$.

By definition, $\Pi$  computes $f^n$ with a \it per-copy \rm error of $\eps$, and thus in particular $\Theta(x,y) = f(x,y)$ with probability $\geq 1-\eps$.
To analyze the information cost of $\Theta$, we write:
\begin{align*} 
& I(\Theta ; x | y) = \E_{i,R_i}[I(\theta ; x | y, R_i)] = \sum_{i=1}^n \frac{1}{n} \cdot I(\Pi;X_i \; | \; Y_i , R_i)  \\
& = \frac{1}{n} \sum_{i=1}^n I(\Pi;X_i \; | \; Y_i , X_{<i} Y_{>i})  = \frac{1}{n} \sum_{i=1}^n I(\Pi;X_i \; | \; X_{<i} Y_{\geq i}) \\
& \leq  \frac{1}{n} \sum_{i=1}^n I(\Pi;X_i \; | \; X_{<i} \Y) =  \frac{1}{n} \cdot I(\Pi; \X \; | \; \Y), 
\end{align*}
where the inequality follows from Lemma \ref{lem_cond_indep}, since $I(Y_{<i} ; X_i | X_{<i}) = 0$
by construction, and the last transition is by the chain rule for mutual information. By symmetry of construction, an analogous argument
shows that $ I(\Theta ; y | x) \leq I(\Pi; \Y \; | \; \X)/n$, and combining these facts gives
\begin{equation} \label{eq_additivity_ic}
\ICprot{\Theta}{\mu} \leq    \frac{1}{n} \left(  I(\Pi; \X \; | \; \Y)  + I(\Pi; \Y \; | \; \X) \right) = \frac{I}{n}. 
\end{equation}
%We remark that the the resource of public coins in such sampling arguments (and other rejection sampling protocol which are widely used in this field) 
%is necessary, and therefore public randomness cannot be eliminated for information purposes, even... 
\end{proof}

\section{Direct Sum, Product, and the Interactive Compression Problem} \label{subsec_ds_dp}

Direct sum and direct product theorems assert a lower bound on the %computational 
complexity of solving $n$ copies of a problem in parallel, in terms of the cost of a single 
copy. %\footnote{information theorists refer to such a result as a ``single-letter characterization", mathematicians often use the term ``tenderization"}. 
Let $f^n$ denote the problem of computing $n$ simultaneous instances of the function $f$ (in some arbitrary computational model for now),
and $C(f)$ denote the cost of solving a single copy of $f$.
The obvious solution to $f^n$ is to apply the single-copy optimal solution $n$ times sequentially and independently to each coordinate, yielding a linear scaling of the 
resources, so clearly $C(f^n)\leq n\cdot C(f)$. The \emph{strong direct sum} conjecture postulates that this naive solution is essentially optimal. 
In the context of randomized  communication complexity, the strong direct sum conjecture informally asks whether it is true that for any 
function $f$ and input distribution $\mu$,

\begin{equation}\label{eq_strong_ds}
\Dmu{f^n}{\mu^n}{\eps} =^{?} \Omega(n)\cdot \Dmu{f}{\mu}{\eps} .
\end{equation}
%When the computational model is randomized (as in distributional/randomized communication complexity), 
More generally, direct sum theorems aim to give an  (ideally linear in $n$, but possibly weaker)  lower bound
on the communication required  for computing $f^n$ with some \emph{constant overall} error $\eps>0$  in terms of the cost of computing a single copy of $f$ with the same (or comparable) 
fixed error.

A \emph{direct product} theorem further asserts that unless 
sufficient resources are provided, the probability of successfully computing all $n$ copies of $f$ will be exponentially small, potentially as low as 
$(1-\eps)^{\Omega(n)}$.  This is intuitively plausible, since the naive solution which applies the best ($\eps$-error) protocol for one copy of $f$ 
independently to each of the $n$ coordinates, would indeed succeed in solving $f^n$ with probability $(1-\eps)^n$. 
Is this naive solution optimal? 

To make this more precise, let us denote by $\suc(\mu,f,C)$ the maximum success probability of a protocol with communication complexity $\leq C$ in computing 
$f$ under input distribution $\mu$. %, and Let $f^n(x_1,\dotsc,x_n, y_1, \dotsc,y_n)$ denote the function that maps its inputs to the $n$ bits $(f(x_1,y_1), ... 
A direct product theorem asserts that any protocol attempting to solve $f^n$ (under $\mu^n$) using some number $T$ of communication bits 
(ideally $T=\Omega(n\cdot C)$), 
will succeed only with exponentially small probability: $\suc(\mu^n,f^n,T) \lesssim (1-\eps)^{\Omega(n)}$. Informally, the strong direct product question asks whether

\begin{equation}\label{eq_strong_dp}
\suc(\mu^n,f^n,o(n\cdot C)) \lesssim^{?} (\suc(\mu,f,C))^{\Omega(n)} .
\end{equation}
Note that \eqref{eq_strong_dp} in particular implies \eqref{eq_strong_ds} when setting $C=\Dmu{f}{\mu}{\eps}$.
%The difference between a direct sum theorem and the (stronger) direct product theorem can be put as follows: A direct sum theorem fixes the success probability 
%(of both the single-copy and parallel computation), and focuses its attention on the increase in resources; A direct product result fixes the resources $T$ for the parallel computation, and focuses on the decay of success probability (hence the terms ``sum" and ``product").
Classic examples of direct product results in complexity theory are Raz's Parallel Repetition Theorem \cite{Raz98, Rao08a} and 
Yao's XOR Lemma  \cite{Yao82} (For more examples and a broader overview of the rich history of direct sum and product theorems
see \cite{JainPP12} and references therein). 
The value of such results to computational complexity
is clear: direct sum and product theorems, together with a lower bound on the (easier-to-reason-about) 
``primitive" problem, yield a lower bound on the composite problem in a ``black-box" fashion (a method also known as \emph{hardness amplification}). For example, the Karchmer-Raz-Wigderson approach for
separating $\mathbf{P}$ from $\mathbf{NC}^1$
can be completed via a (still open) direct sum conjecture for Boolean formulas \cite{KRW95} (after more than a decade, some progress on this conjecture was recently made 
using information-complexity machinery \cite{GMWW14}). Other fields in which direct sums and products have played a central role in proving tight lower 
bounds are streaming \cite{BaryossefJKS04, ST13, MWY13, GO13} and distributed computing \cite{HRVZ13}. \\

Can we always hope for such strong lower bounds to hold?
It turns out that the validity of these conjectures highly depends on the underlying computational 
model, and the short answer is no.\footnote{In the context of circuit complexity, for example, this conjecture 
fails (at least in its strongest form): Multiplying an $n\times n$ matrix by a (worst case) $n$-dimensional vector requires $n^2$ operations, while (deterministic) multiplication 
of $n$ different vectors by the same matrix amounts to matrix-multiplication of two $n\times n$ matrices, which can be done in $n^{2.37} \ll n^3$ operations \cite{Williams12}. }
%What about the communication complexity model? 
In the communication complexity model, this question 
has had a long history and was answered positively for several restricted models of communication 
\cite{Klauck10, Shaltiel03,LSS08,Sherstov12, JainPP12,MWY13,ParnafesRW97}. 
%(For a broader overview of direct sums and 
%products  and their importance in communication complexity we 
%refer the reader to \cite{JainPP12, BRWY12} and references therein).
Interestingly, in the \emph{determistic} communication complexity model, Feder et al. \cite{FederKNN95} showed that 
$$\mathsf{D}(f^n)\geq n\cdot \Omega\left(\sqrt{\mathsf{D}(f)}\right)$$ 
for any two-party Boolean function $f$ (where $\mathsf{D}(f)$ stands for the deterministic communication complexity of $f$), but this proof completely breaks
when protocols are allowed to err. Indeed, in the randomized communication 
model, there is a tight connection between the direct sum question for the function $f$ and its information complexity. 
%After the aforementioned exposition, this should come as no surprise -- 
%Theorem \ref{thm_ic_eq_am_cc} states that the \emph{amortized} communication of $f$ scales like the information cost of a single copy:
%\begin{equation}\label{eq_ds_cc}
%\lim_{n\longrightarrow \infty}\frac{\Dmu{f^n}{\mu^n}{\eps}}{n} = \ICmu{f}{\mu}{\eps},
%\end{equation}
%and therefore the direct sum conjecture \eqref{eq_strong_ds} essentially boils down to the question $\Dmu{f}{\mu}{\eps} =^? O(\ICmu{f}{\mu}{\eps})$ 
%(for the formal argument see [BR 10], Thm...).
%Indeed, the fact that information is fully-additive (which is the ``$\geq$" direction of Theorem \ref{thm_ic_eq_am_cc}) makes it a very natural relaxation for dealing with direct 
%sum and product questions in communication complexity.
By now, this should come as no surprise: 
Theorem \ref{thm_ic_eq_am_cc} %together with Lemma \ref{lem_additivity_of_ic}
asserts that, for large enough $n$, the communication complexity of $f^n$ scales linearly with the (single-copy) information cost of $f$, i.e.  
$\Dmu{f^n}{\mu^n}{\eps} = \Theta\left(n\cdot  \ICmu{f}{\mu}{\eps}\right)$,
and hence the strong direct sum question \eqref{eq_strong_ds} boils down to understanding the relationship between the single-copy measures
$\Dmu{f}{\mu}{\eps}$ and $\ICmu{f}{\mu}{\eps}$. Indeed, it can be formally shown (\cite{BravermanR11}) that the direct sum problem is equivalent
\footnote{The exact equivalence of the direct sum conjecture and Problem \ref{prob_int_compression} holds for \emph{relations} (Theorem 6.6 in \cite{BravermanR11}). 
For total functions, one could argue that the requirement in Problem \ref{prob_interactive_compression} is too harsh as it 
requires simulation of the entire transcript of the  protocol, while in the direct sum context for functions we are merely interested 
in the output of $f$. However, all known compression protocols satisfy the stronger requirement and no separation is known between those techniques.} 
to the following problem of ``one-shot" compression of interactive protocols:  

\begin{problem}[Interactive compression problem, \cite{BBCR}]  \label{prob_int_compression}
\label{prob_interactive_compression}
Given a protocol $\pi$ over inputs $x,y\sim \mu$, with $\|\pi\|=C, \ICprot{\pi}{\mu} = I$, % ($I\ll C$),
what is the smallest amount of communication of a protocol $\tau$ 
which (approximately) simulates $\pi$ (i.e., $\exists \; g $ s.t  $|g(\tau(x,y)) - \pi(x,y)|_1 \leq \delta$ for a small constant $\delta$)?
\end{problem}
In particular, if one could compress any protocol into $O(I)$ bits, this would have shown that $\Dmu{f}{\mu}{\eps} = O\left(\ICmu{f}{\mu}{\eps}\right)$ which 
would in turn imply the strong direct sum conjecture. In fact, the additivity of information cost (Lemma \ref{lem_additivity_of_ic} from Section \ref{sec_ic_additive}) 
implies the following general quantitative relationship between (possibly weaker) interactive 
compression results  and direct sum theorems in communication complexity:

\begin{proposition}[One-Shot Compression implies Direct Sum] \label{cor_compression_ds}
Suppose that for any $\delta>0$ and any given protocol $\pi$ for which $\ICprot{\pi}{\mu} = I$ , $\|\pi\|=C$,
there is a compression scheme that $\delta$-simulates\footnote{The simulation here is in an internal sense, namely, Alice and Bob 
should be able to reconstruct the transcript of the original protocol (up to a small error), based on public randomness and their own private inputs. 
%some publicly known reconstruction function $g$ and the randomness of the protocol.  
See \cite{BRWY12} for the precise definition
and the (subtle) role it plays in context of direct product theorems.}  
$\pi$ using $g_\delta(I,C)$ bits of communication. Then
\[ g_\delta \left(\frac{\Dmu{f^n}{\mu^n}{\eps}}{n} , \Dmu{f^n}{\mu^n}{\eps} \right) \geq \Dmu{f}{\mu}{\eps+\delta}.\]
\end{proposition}

\begin{proof}
Let $\Pi$ be an optimal $n$-fold protocol for $f^n$ under $\mu^n$ with per-copy error $\eps$, i.e., $\|\Pi\| = \Dmu{f^n}{\mu^n}{\eps} := C_n$.
By Lemma \ref{lem_additivity_of_ic} (equation \eqref{eq_additivity_ic}), there is a single-copy $\eps$-error protocol $\theta$ for computing $f(x,y)$ under $\mu$, 
whose information cost 
is at most $\ICprot{\Pi}{\mu^n}/n \leq C_n/n$ (since communication always upper bounds information). 
%The communication complexity of  $\theta$ is quite poor ($C_n$), but 
By assumption of the claim,  $\theta$ can now be $\delta$-simulated using 
$g_\delta(C_n/n ,C_n )$ communication, so as to produce a single-copy protocol with error $\leq \eps+\delta$ for $f$,
and therefore $\Dmu{f}{\mu}{\eps+\delta} \leq g_\delta(C_n/n \;,\; C_n )$.
\end{proof}

%%%%%%%%%%%%%%%%%
%%%%%%%%%%%%%%%%%

%\subsection{How Well Can We Compress?}
The first general interactive compression result was proposed in the seminal work of Barak, Braverman, Chen and Rao
\cite{BBCR}, who showed that any protocol 
$\pi$ can be $\delta$-simulated using  $g_\delta(I,C) = \tilde{O_\delta}(\sqrt{C\cdot I})$ communication
(we prove this result in Section \ref{sec_compression_bbcr}). Plugging this compression result into 
Proposition \ref{cor_compression_ds}, this yields the following 
weaker direct sum theorem:

\begin{theorem}[Weak Direct Sum, \cite{BBCR}]\label{thm_bbcr}
For every Boolean function $f$, distribution $\mu$, and any positive constant $\delta>0$, 
$$ \Dmu{f^n}{\mu^n}{\eps} \geq \tilde{\Omega}(\sqrt{n} \cdot \Dmu{f}{\mu}{\eps+\delta}).$$
\end{theorem}
%where the $\rho$ additive term is due to the fact that compression always introduces an additional small amount of error, as information unlike communication, is an average-case measure. 
Later, Braverman \cite{Bra12} showed that it is always possible to simulate $\pi$ using $2^{O_\delta(I)}$ bits of 
communication. This result is still far from ideal compression ($O(I)$ bits), but it is nevertheless appealing as it show that 
any protocol can be simulated using amount of communication which 
%, thereby exhibiting the first interactive compression scheme which 
depends solely on its information cost, but \emph{independent} of its original communication which may have been arbitrarily 
larger (we prove this result in Section \ref{sec_compression_bra}). Notice that the last two compression results are 
indeed incomparable, since the communication of $\pi$ could be much 
larger than its information complexity (e.g., $C\geq 2^{2^{2^I}}$). The current state of the art for the \emph{general} interactive compression
problem can be therefore summarized as follows: Any protocol with communication $C$ and information cost $I$ 
can be compressed to 
\begin{equation} \label{eq_compresion_best_ub}
g_\delta(I,C) \leq \min \left\{ 2^{O_\delta(I)} \; , \; \tilde{O_\delta}(\sqrt{I\cdot C}) \right\}
\end{equation}
%To prove the strongest form of the direct sum conjecture, we need to be able to compress $\pi$ all the way down to $O(I)$ bits of communication. 
bits of communication.  \\

The above results may seem as a plausible evidence that it is in fact possible to
compress general interactive protocols all the way down to $O(I)$ bits.
Unfortunately, this task turns out to be too ambitious: In a recent breakthrough result,  Ganor, Kol and Raz \cite{GKR14} proved the following
lower bound on the communication of any compression scheme:
\begin{equation} \label{eq_compresion_best_lb}
%\forall \;g_\delta \;\;\;\;\;\;\; 
g_\delta(I,C) \geq \max \left\{ 2^{\Omega(I)} \; , \; \tilde\Omega(I\cdot \log C) \right\}.
\end{equation}
More specifically, they exhibit a Boolean function $f$ which can be solved using a protocol with information cost $I$, 
but cannot be simulated by a protocol $\pi'$ 
with communication cost $<2^{\Omega(I)}$ (a simplified construction and proof was very recently obtained by Rao and Sinha \cite{RaoS15}). 
Since the \emph{communication} of the low information protocol they exhibit is $\sim 2^{2^{I}}$, this also rules out a compression to
$I\cdot o(\log C)$, or else such compression would have produced a too good to be true ($2^{o(I)}$ communication) protocol.
The margin of this text is too narrow to contain the proof of this separation result, 
but it is noteworthy that proving it was particularly challenging: It was shown that essentially all
previously known techniques for proving communication lower bounds apply to information complexity as well
\cite{BW12,KLL}, and hence could not be used to separate information complexity and communication complexity. 
Using (the reverse direction of) Proposition \ref{cor_compression_ds} (see Theorem 6.6 in \cite{BravermanR11}), 
the compression lower bound in \eqref{eq_compresion_best_lb} refutes the strongest possible direct sum \eqref{eq_strong_ds}, but 
leaves open the following gap
\begin{equation} \label{eq_ds_current_gap}
\tilde{\Omega_\delta}\left(\sqrt{n}\right)  \; \leq \; \min_{f} \; \frac{\Dmu{f^n}{\mu^n}{\eps}}{\Dmu{f}{\mu}{\eps+\delta}} \; \leq \; 
O\left(\frac{n}{\log n}\right).
\end{equation}
Notice that this still leaves the direct sum conjecture for randomized communication complexity wide open: It is still conceivable
that improved compression to $g_\delta(I,C)=I\cdot C^{o(1)}$ is in fact possible, and the quest to beat the compression 
scheme of \cite{BBCR} remains unsettled.\footnote{Ramamoorthy and Rao \cite{RR15} recently showed that BBCR's compression scheme can be improved 
when the underlying communication protocol is \emph{asymmetric}, i,e., when Alice reveals much more information than Bob.}
%We briefly discuss recent attempts to do so in Section \ref{sec_improved_compression}. \mnote{O: maybe leave out odometer?}
%However, partial interactive compression results (which allow some dependence on the original communication of $\pi$) still lead to weaker (yet non-trivial) direct sum theorems. The current state-of-the-art result, which partially solves Problem \ref{prob_interactive_compression}, is due to \cite{BBCR}, who showed that 
%$\pi$ can be compressed in to $\tilde{O}(\sqrt{C\cdot I})$ communication. Using Claim \ref{cor_compression_ds}, this yields a weaker %direct sum theorem, namely, that for every constant $\rho>0$,
%\begin{equation}\label{eq_bbcr}
%\Dmu{f^n}{\mu^n}{\eps} = \tilde{\Omega}(\sqrt{n} \cdot \Dmu{f}{\mu}{\eps+\rho})
%\end{equation}
%where the $\rho$ additive term is due to the fact that compression always introduces an additional small amount of error, as information unlike communication, 
%is an average-case measure. 

%As compression of general protocols seems very 
Despite the lack of progress in the general regime, several works showed that it is in fact possible 
to obtain near-optimal compression results in restricted models of communication:
When the input distribution $\mu$ is a \emph{product distribution} ($x$ and $y$ are independent), \cite{BBCR} show a near-optimal compression result, 
namely that $\pi$ can be compressed into $O(I\cdot polylog(C))$ bits.\footnote{\label{footnote_I_ext} These compression results in fact hold for general 
(non-product) distributions as well, when compression is with respect to $I^{ext}$, the external information cost of the original protocol $\pi$ (which may be 
significantly larger than $I$).} 
Once again, using Proposition \ref{cor_compression_ds} this yields the following direct sum theorem:
\begin{theorem}[\cite{BBCR}]\label{eq_bbcr_2}
For every product distribution $\mu$ and any $\delta>0$, 
$$\Dmu{f^n}{\mu^n}{\eps} = \tilde{\Omega}(n \cdot \Dmu{f}{\mu}{\eps+\delta}).$$
\end{theorem}
Improved compression results  %($O(I^2 \cdot \log\log(C))$) 
were also proven for \emph{public-coin protocols} (under arbitrary distributions) \cite{BBKLSV13, BMY14}, 
and for bounded-round protocols, leading to near-optimal direct sum theorems in corresponding communication 
models. We summarize these results in Table \ref{tab:results}.  
%intuitively, compressing deterministic protocols is an easier task because
%``in this case $p=\Pr[M=1 | x, \text{history} ] \in \{0,1\}$, and in this regime the loss from Pinsker's inequality can be avoided:
%$\Div{p}{q} \approx \|p-q\|_1$.  

%%%%%%%%%%%%%%%%%%%%%%%%%%%%%%%%%%%%%%%%%%%%%%%%%%%%%%%%%%%%%%%%%%%%%%%%%%%%%%%%%%%%%%%%
%%%%%%%%%%%%%%%%%%%%%%%%%%%%%%%%%%%%%%%%%%%%%%%%%%%%%%%%%%%%%%%%%%

\begin{table*}[!ht]
\centering
\scalebox{1}{
\begin{tabular}{|c|c|c|}
\hline
%& \multicolumn{2}{|c|}{With duplication} & \multicolumn{2}{|c|}{Without duplication}\\
%\hline
Reference & Regime   & Communication Complexity \\ \hline 
%\cite{ams}   & $O(n^{1-1/p}\epsilon^{-2}\log(Mn))$ & does not work in turnstile model\\
\multirow{ 2}{*}{} \cite{HJMR07} &  $r$-round protocols, & $I + O(r)$ \\
                                                            &  product distributions\footnoteref{footnote_I_ext} & \\ 
\cite{BravermanR11, BRWY13b} &  $r$-round protocols & $I +O\left(\sqrt{r\cdot I}\right) + O(r\log 1/\delta)$ \\
\cite{BMY14} (improved \cite{BBKLSV13}) & Public coin protocols     & $O(I^2\cdot \log\log(C)/\delta^2)$ \\
\cite{BBCR} & Product distributions\footnoteref{footnote_I_ext}%\footnote{This result in fact holds for general (non-product) distributions as well, when the compression is 
%with respect to $I^{ext}$, the external information cost of the original protocol $\pi$.}  
& $O(I\cdot \poly\log(C)/\delta)$ \\
\cite{Bra12, BBCR} & {\bf General protocols} & $\min \{ 2^{O(I/\delta)} \; , \; O(\sqrt{I\cdot C} \cdot \log(C)/\delta) \}$ \\
\cite{GKR14,RaoS15} & {\bf Best lower bound} & $\max \{ 2^{\Omega(I)} \; , \; \Omega(I\cdot \log(C))$ \\
\hline
\end{tabular}
}
\caption{
Best to date compression schemes, for various regimes. % (and simulation error $\delta=O(1)$). 
%For the general case, 
%last two results show that \cite{Bra12}'s $2^{O(I)}$ compression scheme is tight for compression depending solely on $I$. 
Notice that in the general regime (last two columns), in terms of dependence on the original communication $C$, the gap is still very large ($\Omega(\log C)$ vs.\ $\tilde{O}(C^{1/2})$).
}\label{tab:results}
\end{table*}

%%%%%%%%%%%%%%%%%%%%%%%%%%%%%%%%%%%%%%%%%%%%%%%%%%
%%%%%%%%%%%%%%%%%%%%%%%%%%%%%%%%%%%%%%%%%%%%%%%%%%
%%%%%%%%%%%%%%%%%%%%%%%%%%%%%%%%%%%%%%%%%%%%%%%%%%

\subsection{Harder, better, stronger: From direct sum to direct product} \label{subsec_dp}

As noted above,
direct sum theorems  such as Theorems \ref{thm_ic_eq_am_cc}, \ref{thm_bbcr}  and \ref{eq_bbcr_2}
are weak in that they merely assert that 
attempting to solve $n$ independent copies of $f$ using less than some number $T$ of resources, would fail with some \emph{constant}
overall probability (($\suc(\mu^n,f^n,$ $o(\sqrt{n\cdot C}))\leq \eps$ in the general case, and
$\suc(\mu^n,f^n,o(n\cdot C)) \leq \eps$ in the product case, where $C=\Dmu{f}{\mu}{\eps}$). This is somewhat unsatisfactory, since the naive 
solution that applies the single-copy optimal protocol independently to each copy 
has only exponentially  small success probability in solving all copies correctly. 
Indeed, some of the most important applications of hardness amplification require amplifying the error parameter (e.g., the usage
of parallel repetition in the context of the PCP theorem).

%Obtaining strong direct product theorems was a central and fruitful research endeavor in the past two decades
As mentioned before, many
direct product theorems were proven in limited communication models (e.g. Shaltiel's Discrepancy bound
\cite{Shaltiel03,LSS08} which was extended to the generalized discrepancy bound \cite{Sherstov12},
 Parnafes, Raz, and Wigderson's theorem for communication forests \cite{ParnafesRW97},  Jain's theorem \cite{jain11}  
 for simultaneous communication  and \cite{JY12}'s direct product in terms of the ``smooth rectangle bound" to mention a few), 
 but none of them applied to general functions and communication protocols.
In a recent breakthrough work, Jain, Pereszl\'enyi and Yao used an information-complexity based approach to  prove 
 a strong direct product theorem for any function (relation) in the bounded-round communication model.
 %, which can be viewed as a sharpening of the direct sum theorem for the bound-round model of \ref{BravermanR11} :

\begin{theorem}[\cite{JainPP12}] \label{thm_dp_jpy} 
Let $\suc_{r}(\mu,f,C)$ denote the largest success probability of an $r$-round protocol with communication at most $C$,
and  suppose that $\suc_{r}(\mu,f,C) \leq \frac{2}{3}$.
If $T = o\left(\left(\frac{C}{r} - r\right)\cdot n\right)$, then $\suc_r(\mu^n,f^n,T)\leq \exp(-\Omega(n/r^2))$.
\end{theorem}
This theorem can be essentially viewed as a sharpening of the direct sum theorem of Braverman and Rao for bounded-round communication 
\cite{BravermanR11}.
This bound was later improved by Braverman et. al who showed that $\suc_{r/7}(\mu^n,f^n,o((C- r \log r) \cdot n))\leq \exp(-\Omega(n))$,
thus settling the strong direct product conjecture in the bounded round regime. 
The followup work of \cite{BRWY12} took this approach one step further, obtaining the first direct product theorem for \emph{unbounded-round}
randomized communication complexity,
thus sharpening the direct sum results of \cite{BBCR}.

\begin{theorem}[\cite{BRWY12}, informally stated] \label{thm_dp_main}
For any two-party function $f$ and distribution $\mu$ such that $\suc(\mu, f, C) \leq \frac{2}{3}$, the following holds:
%For any two-party function $f$ the following holds:
\begin{itemize}
\item If $T \log^{3/2} T = o(C\cdot\sqrt{n})$, then $\suc(\mu^n,f^n,T) \leq \exp \left(-\Omega(n) \right)$.
\item If $\mu$ is a product distribution, and $T \log^2 T =o( C\cdot n)$, then $\suc(\mu^n, f^n, T) \leq \exp(-\Omega(n))$.
%\item For bounded-round protocols using only $r$ rounds, we show that 
%if $\suc_{7r}(\mu,f,C) \leq \frac{2}{3}$ and $T \leq (C- \Omega(r \log r)) \cdot n$ then
%$\suc_{r}(\mu^n,f^n,T)\leq\exp(-\Omega(n))$  (where $\suc_r(\mu,f,C)$ denotes the maximum success probability of an $r$-round communication 
%protocol using $\leq C$ bits). 
\end{itemize}
\end{theorem}

One appealing corollary of the second proposition is that, under the \emph{uniform} distribution, two-party interactive computation 
cannot be ``parallelized", in the sense that the best protocol for solving $f^n$ (up to polylogarithmic factors), is to apply the single-coordinate 
optimal protocol independently to each copy, which almost matches the above parameters.  \\

%While the works above focus on obtaining generic direct product theorems in communication complexity,
%other works focused on proving direct product theorems in terms of weaker complexity measures of the underlying function, such as the discrepancy 
%$\mathsf{disc}_\mu(f)$ of the function (\cite{LSS08}) or the (stronger) smooth rectangle bound $\mathsf{srect}_\mu(f)$ of \cite{JK10} (informally, \cite{JY12} 
%show that any protocol attempting to compute $f^n$ under $\mu^n$ using $\ll n\cdot \mathsf{srect}_\mu(f)$ 
%communication, will succeed with probability only $2^{-\Omega(n)}$).  

The high-level intuition behind the proofs of Theorems \ref{thm_dp_jpy} and \ref{thm_dp_main} follows the direct sum approach of \cite{BBCR}
(Proposition \ref{cor_compression_ds} above): Suppose, towards contradiction, that the success probability of an $n$-fold protocol using $T$ bits of communication 
in computing $f^n$ under $\mu^n$ is larger than, say, $\exp(-n/100)$. We would like to ``embed" a \emph{single-copy} $(x,y)\sim \mu$ into this $n$-fold protocol, 
thereby producing a \emph{low information} protocol ($\leq T/n$ bits), and then use known compression schemes to compress this
protocol, eventually obtaining  a protocol with communication ($<C$), and a too-good-to-be-true success probability ($>2/3$), contradicting the assumption that
$\suc(\mu, f, C) \leq \frac{2}{3}$.
The main problem with employing the \cite{BBCR} approach and embedding a single-copy $(x,y)$ into $\pi$ using the sampling argument in
Lemma \ref{lem_additivity_of_ic}, is that it would produce a single-copy protocol $\theta(x,y)$ whose success probability is no better than that
of $\pi$ ($\exp(-n/100)$) while we need to produce a single-copy protocol with success $>2/3$ in order to achieve the above contradiction.

Circumventing this major obstacle is inspired by the idea of repeated conditioning which first appeared the parallel repetition theorem \cite{Raz98}:
Let $\cW$ be the event that $\pi$ correctly computes $f^n$, and $\cW_i$ denote the event that the protocol correctly computes 
the $i$'th copy $f(x_i,y_i)$. Let $\pi(\cW)$ denote the probability of $\cW$, and $\pi(\cW_i|\cW)$ denote the conditional probability of the event $\cW_i$ 
given $\cW$ (clearly, $\pi(\cW_i|\cW)=1$). The idea is to show that if $\pi(\cW) \geq \exp(-n/100)$ and $\|\pi\|\ll T$ (for the appropriate choice of $T$ which 
is determined by the best compression scheme), then $(1/n) \sum_{i=1}^n \pi(\cW_i|\cW) < 1$, which is 
a contradiction. In other words, if one could simulate the message distribution of the conditional distribution $(\pi|\cW)_i$ (rather than the distribution of 
$\pi(x_i,y_i)$) using a low information protocol, then (via compression) one would obtain a protocol $\theta(x_i,y_i)$ with \emph{constant} success probability, 
as desired. 

The guiding intuition for why this approach makes sense, is that conditioning a random variable on a ``large" event $\cW$ does not change its original distribution 
too much:
\begin{align*}
&\Div{X_1Y_1, X_2Y_2, \ldots , X_nY_n \; | \cW}{X_1Y_1, X_2Y_2, \ldots , X_nY_n}  = \Div{\X\Y|\cW}{\X\Y} \\
=& \E\left[ \log \frac{\pi(\X\Y|\cW)}{\pi(\X\Y)} \right] 
 \leq \E\left[ \log \frac{\pi(\X\Y)}{\pi(\X\Y)\pi(\cW)} \right]  = \frac{1}{\log(\pi(\cW))} \leq \frac{n}{100}
\end{align*}
since $\pi(\cW) \geq \exp(-n/100)$, which means (by the chain rule and independence of the $n$ copies) that the distribution of an \emph{average} input pair $(X_i,Y_i)$ conditioned on $\cW$ is $(1/100)$-close
to its original distribution $\mu$, and thus implies that at least the \emph{inputs} to the ``protocol" $(\pi|\cW)_i$ can be approximately sampled correctly 
(using correlated sampling \cite{Holenstein07}). The heart of the problem, however, is that %after conditioning on the event $W$, 
$(\pi|\cW)_i$ is no longer a communication protocol. To see why, consider the simple protocol $\pi$ in which Alice simply ``guesses" Bob's bit $x$,
and $\cW $ being the event that her guess is correct. Then simulating $(\pi|\cW)$ requires Alice to know Bob's input $y$, which Alice doesn't have!
This example shows that it is impossible to simulate the message distribution of $(\pi|\cW)_i$ exactly. The main contribution of 
Theorem \ref{thm_dp_main} (and Theorem \ref{thm_dp_jpy} in the bounded-round regime) is showing that it is nevertheless possible to 
\emph{approximate} this conditional distribution using 
an actual communication protocol, which is statistically close to a low-information protocol:

\begin{lemma}[Claims 26 and 27 from \cite{BRWY12}, informally stated] \label{lem:brwy}
There is a protocol $\theta$ taking inputs $x,y\sim\mu$ so that the following holds:
\begin{itemize}
\item $\theta$ publicly chooses a uniform $i \in [n]$ independent of $x,y$, and $R_i$ which is part of the input to $\pi$
(intuitively, $R_i$ determines the ``missing" inputs $x_{-i}, y_{-i}$ of $\pi$ as in Lemma \ref{lem_additivity_of_ic}).
\item $\E_i \left[| (\theta|R_i) - (\pi|R_i \cW)_i |\right] \leq 1/10$  (that is, $\theta$ is close to the distribution $(\pi|\cW)_i$ for average $i$).
%\item $Rounds(\theta)=Rounds(\pi)$.
\item $
\E_{i} \left[  I_{\pi|\cW}(X_i; \Pi | Y_i R_i) +I_{\pi|\cW}(Y_i; \Pi | X_i R_i) \right] \leq 4\| \pi \|/n$ (that is, the information cost of the distribution $(\pi|\cW)_i$ is low).
\end{itemize} 
\end{lemma}

The main challenge in proving this theorem is in the choice of the public random variable $R_i$, which enables relating
 the information of the protocol $\theta$ to that of $(\pi|\cW)$ \emph{even under the conditioning on $\cW$}. This technically-involved  
 argument is a ``conditional" analogue of Lemma \ref{lem_additivity_of_ic} (for details see \cite{BRWY12}).
Note that the last proposition of Lemma \ref{lem:brwy} only guarantees that the information cost of the transcript under the  distribution $(\pi | \cW)$ is low 
(on an average coordinate), while we need this property to hold for the simulating protocol $\theta$, in order to apply the compression schemes of \cite{BBCR} 
which would finish the proof.  Unfortunately, a protocol $\pi$ that is statistically close to a low-information distribution needs not be a low-information protocol itself: 
Consider, for example, a protocol $\pi$
where with probability $\delta$ Alice sends her input $X\in\{0,1\}^n$ to Bob, and with probability $1-\delta$ she sends a random string. Then $\pi$
is $\delta$-close to a $0$-information protocol, but has information complexity of $\approx \delta \cdot n$, which could be arbitrarily high. 
\cite{BRWY12} circumvented this problem by showing that the necessary compression schemes of \cite{BBCR} are ``smooth" in the sense that they also work
for protocols that are merely close to having low-information. In a followup work, Braverman and Weinstein exhibited a general 
technique for converting protocol which are statistically-close to having low information into actual low-information protocols (see Theorem 3 in \cite{BW14}), 
which combining Lemma 
\ref{lem:brwy} also led to a strong direct product theorem in terms of information complexity, sharpening the ``Information=Amortized Communication" 
Theorem of Braverman and Rao:

\begin{theorem}[\cite{BW14}, informally stated]\label{thm:dp_for_ic}
%Let $\suci(\mu, f, I))$ denote the largest success probability of a protocol whose information cost is at most $I$
Suppose that  $\ICmu{f}{\mu}{2/3} = I$, i.e., solving a single copy of $f$ with probability $2/3$ under $\mu$ requires $I$ bits of information.
If
%Let $f$ be a 2-party Boolean function. There are universal constants $\alpha, \beta>0$ such that if $\gamma = \beta(1-\suci(\mu, f, I))/2$, 
$T\log(T) = o(n \cdot I)$, then   $\suc(\mu^n,f^n,T) \leq \exp \left(- \Omega(n)  \right).$
\end{theorem}
In fact, this theorem shows that the direct sum and  product conjectures in randomized communication complexity are equivalent
(up to polylogarithmic factors), and they are both equivalent to one-shot interactive compression, in the quantitative sense of Proposition 
\ref{cor_compression_ds} (we refer the reader to \cite{BW14} for the formal details).

\section{State of the Art Interactive Compression Schemes} \label{sec_compression}

In this section we present the two state-of-the-art compression schemes for unbounded-round communication protocols, 
the first due to Barak et al., and the second due to Braverman \cite{BBCR, Bra12}. 
As mentioned in the introduction, a natural idea for compressing a multi-round protocol is to try and compress each round separately, 
using ideas from the transmission (one-way)
setup \cite{huffman1952method, HJMR07, BravermanR11}. 
Such compression suffers from one fatal flaw: It would inevitably require sending at least $1$ bit of communication at 
each round, while the information revealed in each round may be $\ll 1$ (an instructive example is the protocol in which Alice sends Bob, at each round of the protocol,
an independent coin flip which is $\eps$-biased towards her input $X\sim Ber(1/2)$, for $\eps \ll 1$).
%In principle, one could try to apply round-by-round message compression to compress interactive protocols, i.e., reducing multi-round protocol compression
%to the classical data compression setup on Shannon and Huffman [?] . This approach, however, is bound to fail due to the following reason: 
%Compressing each round separately clearly require sending at least $1$ bit of communication per round, while individual messages  of $\pi$ may 
%(and are in fact likely to) convey $\ll 1$ bits of information. If $I \ll rounds(\pi)$, this approach would still yield a very large communication  protocol,
%compared to the information it contains (consider the instructive example where each bit of the protocol... ). 
Thus any attempt to implement the compression on a round- by-round basis is hopeful only when the number of rounds 
is bounded but is doomed to fail in general (indeed, this is the essence of the bounded-round compression schemes of \cite{BravermanR11, BRWY13b}).

The main feature of the compression results we present below is that they do not depend on the number of rounds of the underlying 
protocol, but only on the overall communication and information cost.

\subsection{Barak et al.'s compression scheme}\label{sec_compression_bbcr}
%The first interactive compression result was due to Barak, Braverman, Chen and Rao, who showed the following:
\begin{theorem}[\cite{BBCR}] \label{thm_bbcr_internal}
Let $\pi$ be a protocol executed over inputs $x,y \sim \mu$, and suppose $\ICprot{\pi}{\mu} = I$ and $\|\pi\| = C$.  
Then for every $\eps>0$, there is a protocol $\tau$ which $\eps$-simulates $\pi$, where 
\begin{align}\label{eq_cc_final}
\|\tau\| =  O\left( \sqrt{C\cdot I} \cdot (\log(C/\eps)/\eps)\right).
\end{align}
\end{theorem}
%It is instructive to consider the care where $I \ll C$, in which case the above compression scheme indeed significantly reduces the communication of $\pi$.
\begin{proof}
The conceptual idea underlying this compression result is using public randomness %using correlated sampling 
to \it avoid communication by trying to guess what the other player is about to say\rm. 
Informally speaking, the players will use shared randomness to sample (correlated) \it full paths \rm of the protocol tree, 
according to their private knowledge: Alice has the ``correct" distribution on nodes that she owns in the tree (since conditioned on reaching these nodes, the next 
messages only depend on her input $x$), and will use her  ``best guess" (i.e., her prior distribution on Bob's next message, under $\mu$, her input $x$ and the history
of messages) to sample messages at 
nodes owned by Bob. Bob will do the same on nodes owned by Alice. This ``guessing" is done in a correlated way using public randomness (and no communication 
whatsoever (!)), in a way that  guarantees that if the player's guesses are close to the correct distribution, then the probability that they 
 sample the same bit is large. %Intuitively, the fact that the information cost os low would imply that these distributions are indeed close.
 
The above step gives rise to two paths, $P_A$ and $P_B$ respectively. In the the next step, the  players will use (mild) communication to find all inconsistencies among 
 $P_A$ and $P_B$ and correct them one by one (according to the ``correct" speaker). By the end of this process, the players obtain a consistent path which 
 has the correct distribution $\Pi(x,y)$.  
Therefore, the overall communication of the simulating protocol would be comparable to the number of mistakes between $P_A$ and $P_B$ (times the communication cost
of fixing each mistake). Intuitively, the fact that 
$\pi$ has low information will imply that the number of inconsistencies is small, 
as inconsistent samples on a given node typically occur when the ``receiver's" prior distribution is far from the ``speaker's" correct distributions, which will in turn imply 
that this bit conveyed a lot of information to the receiver (Alas, we will see that if the information revealed by the $i$'th bit of $\pi$ is $\eps$, then the probability of 
making a mistake on the $i$'th node is $\approx \sqrt{\eps}$,
and this is the source of sub-optimality of the above result. We discuss this bottleneck at the end of the proof). 

We now sketch the proof more formally (yet still leaving out some minor technicalities).  
Let $\Pi=M_1,\ldots,M_{C}$ denote the transcript of $\pi$.
Each node $w$ at depth $i$ of the protocol tree of $\pi$ is associated 
with two numbers, $p_{x,w}$ and $p_{y,w}$, describing the probability (according to each player's respective ``belief") that conditioned on reaching $w$, the next bit 
sent in $\pi$ is ``$1$" (the right child of $w$). That is, 
\begin{align}\label{eq_transition_probabilities}
p_{x,w} := \Pr[M_{i} = 1 \; | \; xr M_{<i}=w] \;\;\;\; \text{, and} \;\;\;\; p_{y,w} := \Pr[M_{i}=1 \; | \; yr, M_{<i}=w].
\end{align}
Note that if $w$ is owned by the Alice, then $p_{x,w}$ is exactly the correct probability with which the $i$-th bit is transmitted in $\pi$, conditioned that $\pi$ has reached $w$.
%Hence, in this terminology, the player's goal is to sample a path $\ell$ in the protocol tree of $\pi$, according to the distribution:  \mnote{fix}
%\begin{align} \label{def_path_prob}
%\Pr[\ell] = \Pi_{v \text{is even} p_{x,v}(\ell_)}
%\end{align}

In the simulating protocol $\tau$, the players first sample, without communication and using public randomness, a uniformly random number $\rho_w$ in the interval $[0,1]$, 
for every node $w$ of the protocol tree\footnote{Note that there are exponentially many nodes, but the communication model does not charge for local computations or the amount of shared 
randomness, so these resources are indeed ``for free".}.  For simplicity of analysis, in the rest of the proof we assume the public randomness is fixed to the vale $R=r$. 
%(though the protocol may still use private randomness).
Alice and Bob now privately construct the following respective trees $\cT_A,\cT_B$: For each node $w$, Alice includes
the right child of $w$ in $\cT_A$ iff $p_{w,x} < \rho_w$, and the left child (``$0$") otherwise. Bob does the same by including the right child of $w$ in $\cT_B$ iff $p_{w,y} < \rho_w$.
%Alice will set the next bit of her path $P^A_v=``1"$ iff $\rho < p_{x,v}$ (and $b_v=``0" otherwise$), and Bob will set 
%$P^B_v=``1"$ iff $\rho < p_{y,v}$ . Notice that, since $\rho_v$ is uniform 
%in $[0,1]$, the probability that Alice samples $P^A_v=``1"$ is exactly $p_{v,x}$, which is the correct probability whenever Alice owns $v$ -- otherwise, an analogues 
%argument for Bob shows that he will have the correct sample,  according to $p_{y,v}$. Note that $P^A$ and $P^B$ are well defined by the above process.

The trees $\cT_A$ and $\cT_B$ define a unique path $\ell =m_1,\ldots,m_{C}$ of $\pi$, by combining outgoing edges from $\cT_A$ in nodes owned by Alice, 
and edges from $\cT_B$ in nodes owned by Bob.  Note that $\ell$ has precisely the desired distribution of $\Pi(X,Y)$.
To identify $\ell$, the players will now find the inconsistencies among $\cT_A$ and $\cT_B$ and correct them one by one. 

We say that a \it mistake \rm occurs in level $i$ if the outgoing edges of $m_{i-1}$ in $\cT_A$ and $\cT_B$ are inconsistent.  
Finding the (first) mistake of $\tau$ amounts to finding the first differing index among two $C$-bit strings (corresponding to the paths $P_A$ and $P_B$ 
induced by $\cT_A$ and $\cT_B$).   
Luckily, there is a randomized protocol which accomplishes this task with high probability ($1-\gamma$) using only $O(log(C/\gamma))$ bits of 
communication, using a clever ``noisy" binary search due to Feige et al. \cite{FeigePRU94}. Since errors accumulate over $C$ rounds and we are aiming for  
an overall simulation error of $\eps$, we will set $\gamma \approx \eps/C$,  thus the cost of fixing each inconsistency remains $O(\log(C/\eps))$ bits.
The expected communication complexity of $\tau$ (over $X,Y,R$) is therefore
\begin{align}\label{eq_exp_cc_tau}
\E[\|\tau\|] = \E[\# \text{ mistakes of } \tau] \cdot O(\log(C/\eps)).
\end{align}

Though we are not quite done, one should appreciate the simplicity of analysis of the cost of this protocol.
The next lemma completes the proof,  asserting that the expected number of mistakes $\tau$ makes is not too large:
\begin{lemma}\label{lem_exp_mistakes}
$\E[ \# \text{ mistakes of } \tau] \leq \sqrt{C\cdot I}$.
\end{lemma}

Indeed, substituting the assertion of Lemma \ref{lem_exp_mistakes} into \eqref{eq_exp_cc_tau}, we conclude that the expected communication 
complexity of $\tau$ is $O(\sqrt{C\cdot I} \cdot \poly\log(C/\eps))$, and a standard Markov bound yields the bound in \eqref{eq_cc_final} and therefore 
finishes the proof of Theorem \ref{thm_bbcr_internal}.

\begin{proof}[Proof of Lemma \ref{lem_exp_mistakes}]
Let $\cE_i$ be the indicator random variable denoting whether a mistake has occurred in step $i$ of the protocol tree of $\pi$. Hence the expected number of 
mistakes is $\sum_{i=1}^C \cE_i$. We shall bound each term $\E[\cE_i]$ separately. 
By construction,  a mistake at node $w$ in level $i$ occurs exactly when either $p_{x,w} < \rho_w < p_{y,w}$ or $p_{y,w} < \rho_w < p_{x,w}$. Since $\rho_w$ was uniform
in $[0,1]$, the probability of a mistake is
\[ |p_{x,w} - p_{y,w}| = |(M_i | x, r, M_{<i} = w)  - (M_i | y, r, M_{<i} = w)|, \] % = |\Pr[M_i = 1 | x, r, m_{<i}]  - \Pr[M_i = 1 | y, r, m_{<i}]|,  \]
where the last transition is by definition of $p_{x,w}$ and $p_{y,w}$. Note that, by definition of a protocol, if $w := m_{<i}$ is owned by Alice, then 
$M_i | xyrm_{<i}] = M_i | xyrm_{<i}$ and if it is owned by Bob, then  $M_i | y, r, m_{<i} = M_i | x, y, r, m_{<i}$. We therefore have
\begin{align}
& \E[\cE_i] =  \E_{xym_{<i}\sim \pi}[|(M_i | xrm_{<i})  - (M_i | yrm_{<i})|] \nonumber \\
& \leq \E_{xym_{<i}\sim \pi}\left [\max\{ |(M_i | xyrm_{<i})  - (M_i | xrm_{<i})| \;, \; |(M_i | xyrm_{<i})  - (M_i | yrm_{<i})| \}\right ] \nonumber \\ 
& \leq \E_{xym_{<i}\sim \pi}\left[\sqrt{ \Div{M_i | xyrm_{<i}}{M_i | xrm_{<i}} +  \Div{M_i | xyrm_{<i}}{M_i | yrm_{<i}} } \right]      \label{eq_pinsker} \\ 
& \leq \sqrt{ \E_{xym_{<i}\sim \pi} \left[\Div{M_i | xyrm_{<i}}{M_i | xrm_{<i}} +  \Div{M_i | xyrm_{<i}}{M_i | yrm_{<i}} \right]}     \label{eq_convexity} \\ 
& = \sqrt{I(M_i ; X | M_{<i} R Y) + I(M_i ; Y | M_{<i} R X)}
\end{align}
where transition \eqref{eq_pinsker} follows from Pinsker's inequality (Lemma \ref{lemma:pinsker}), transition \eqref{eq_convexity} follows from
the convexity of $\sqrt{\cdot}$, and the last transition is by Proposition \ref{prop_dim_IC}. \\

Finally, by linearity of expectation and the Cauchy-Schwartz inequality, we conclude that
\begin{align*}
& \E\left[\sum_{i=1}^{C}\cE_i\right] \leq \sum_{i=1}^C \sqrt{I(M_i ; X | M_{<i} R Y) + I(M_i ; Y | M_{<i} R X)}  \nonumber \\
& \leq \sqrt{ \left(\sum_{i=1}^C 1\right ) \cdot \left(\sum_{i=1}^C  I(M_i ; X | M_{<i} R Y) + I(M_i ; Y | M_{<i} R X)\right)}  \nonumber \\
& = \sqrt{C \cdot I}
\end{align*}
where the last transition is by the chain rule for mutual information.
\end{proof}

\end{proof}
\smallskip

A natural question arising from the above compression scheme is whether the analysis in Lemma \ref{lem_exp_mistakes} is tight.  
Unfortunately, the answer is yes, as demonstrated by the following 
example: Suppose Alice has a single uniform bit $X\sim Ber(1/2)$, and consider the $C$-bit protocol in which Alice sends, at each round $i$, an independent
sample $M_i$ such that 

$$
M_i \sim \left\{ \begin{array}{rl}
 Ber(1/2+\eps) &\mbox{ if $x=1$}\\
   Ber(1/2-\eps) &\mbox{ if $x=0$}
       \end{array} \right.
$$
for $\eps = 1/\sqrt{C}$. Since Bob has a perfectly uniform prior on $X$, a direct calculation shows that in this case 
$I(M_i;X|M_{<i}) \leq I(M_i;X) = \Div{Ber(1/2+\eps)}{Ber(1/2)}= O(\eps^2)$, while the probability of making a ``mistake" at step $i$ of the simulation above
is the total variation distance $|Ber(1/2+\eps) - Ber(1/2)| \approx \eps.$ Therefore, the expected number of mistakes conditioned on, say, $x=1$, is $C\cdot \eps = \sqrt{C}$, 
by choice of $\eps=1/\sqrt{C}$. I.e., this example shows that both Pinsker's and the Cauchy-Schwartz inequalities are tight in the extreme case where 
each of the $C$ bit of $\pi$ reveals $\approx I/C$ bits of information. In the next section we present a different compression scheme which can do better in this
regime, at least when $I$ is much smaller than $C$.

%%%%%%%%%%%%%%%%%%%%%%%%%%%%%%%%%%%%%%%%%%%%%%%%%%%%%%%%%%%%%%%%%%%%%%%%%%%%%%%%%%%%%%%%
%%%%%%%%%%%%%%%%%%%%%%%%%%%%%%%%%%%%%%%%%%%%%%%%%%%%%%%%%%%%%%%%%%%%%%%%%%%%%%%%%%%%%%%%
%%%%%%%%%%%%%%%%%%%%%%%%%%%%%%%%%%%%%%%%%%%%%%%%%%%%%%%%%%%%%%%%%%%%%%%%%%%%%%%%%%%%%%%%
%%%%%%%%%%%%%%%%%%%%%%%%%%%%%%%%%%%%%%%%%%%%%%%%%%%%%%%%%%%%%%%%%%%%%%%%%%%%%%%%%%%%%%%%

\subsection{Braverman's compression scheme} \label{sec_compression_bra}

\begin{theorem}[\cite{Bra12}] \label{thm_bra_compression}
Let $\pi$ be a protocol executed over inputs $x,y \sim \mu$, and suppose $\ICprot{\pi}{\mu} = I$.
Then for every $\eps>0$, there is a protocol $\tau$ which $\eps$-simulates $\pi$, where  $\|\tau\| =   2^{O(I/\eps)}$.
\end{theorem}

\begin{proof}
To understand this result, it will be useful to view the interactive compression problem
as the following correlated sampling task: 
%Let $\pi$ be a protocol such that $\ICprot{\pi}{\mu} = I$. 
Denote by $\pi_{xy}$ the distribution of the transcript $\Pi(x,y)$, 
and by $\pi_x$ (resp. $\pi_y$) the conditional marginal distribution $\Pi | x$ ($\Pi | y$) of the transcript from Alice's (Bob's) point of view (for notational ease, 
the conditioning on the public randomness $r$ of the protocol is included here implicitly. Note that in general $\pi$ is still randomized even conditioned on $x,y$, 
since it may have private randomness). By the product structure of communication protocols, %$\pi_{xy}(\ell) = \prod_{i=1}^C \Pr[w_i | x,y,r, w_{<i}]$, and therefore
the probability of reaching a leaf (path) $\ell\in \{0,1\}^C$ of $\pi$ is
\begin{align}\label{eq_dist_pi_correct}
\pi_{xy}(\ell) = p_{x}(\ell) \cdot p_{y}(\ell)
\end{align}
where $p_{x}(\ell) = \prod_{w \subseteq \ell, \text{$w$ odd}} p_{x,w}$ is the product of the transition probabilities defined 
in  \eqref{eq_transition_probabilities} on the nodes owned by Alice along the path from the root to $\ell$, and
$\pi_{y}(\ell)$ is analogously defined on the even nodes. Thus, the desirable distribution from which the players wish to jointly sample, 
decomposes to a natural product distribution \footnote{As we shall see, the rejection sampling approach of the compression protocol below 
crucially exploits this product structure of the target distribution, and it is curious to
note this simplifying feature of interactive compression as opposed to general correlated sampling tasks.}.
Similarly, 
\begin{align}\label{eq_dist_pi_marginal}
\pi_{x}(\ell) = p_{x}(\ell) \cdot q_{x}(\ell) \;\;\;\;\;\; \text{and} \;\;\;\;\; \pi_{x}(\ell) = q_{y}(\ell) \cdot p_{y}(\ell)
\end{align}
where $q_{x}(\ell) = \prod_{w \subseteq \ell, \text{$w$ even}} p_{x,w}$ is Alice's prior ``belief" on the \it even nodes \rm owned by Bob along the path to $\ell$
(see \eqref{eq_transition_probabilities}), and $q_{y}(\ell) = \prod_{w \subseteq \ell, \text{$w$ odd}} p_{x,w}$ is Bob's prior belief 
on the  odd nodes owned by Alice. Thus, the player's goal is to sample $\ell \sim \pi_{x,y}$, where  Alice has the correct distribution on odd nodes (and only an estimate
on the odd ones),  and Bob has the correct distribution on even nodes (and an estimate on the even ones). \\

We claim that the information cost of $\pi$ being low ($I$) implies that Alice's prior ``belief" $q_x$ on the even nodes owned by Bob, 
is ``close" to the true distribution $p_y$ on these nodes (and vice versa for $q_y$ and $p_x$ on the odd nodes). % when conditioned on Bob's input $y$ (and vice versa).
%Given a protocol $\pi$ with inputs $(x,y)\sim \mu$, the goal of the players will be to jointly sample a 
%leaf (path) of the protocol according to the correct distribution $\ell \sim \pi_{xy} = p_x(\ell) p_y(\ell)$. 
To see this, recall the equivalent interpretation of mutual information in terms of KL-divergence:
\begin{align}\label{eq_ic_div_pi}
& I = I(\Pi ; X | Y) + I(\Pi ; Y | X) = \E_{(x,y)\sim \mu } \left[\Div{\pi_{xy}}{\pi_y} + \Div{\pi_{xy}}{\pi_x} \right] \nonumber \\
& = \E_{x,y,\ell \sim \pi_{x,y}} \left[ \log\frac{\pi_{xy}(\ell)}{\pi_y(\ell)}  \; + \log\frac{\pi_{xy}(\ell)}{\pi_x(\ell)} \;  \right] 
 = \E_{x,y,\ell \sim \pi_{x,y}} \left[ \log\frac{p_x(\ell)}{q_y(\ell)}  \; + \log\frac{p_{y}(\ell)}{q_x(\ell)} \;  \right] ,
\end{align}
where the last transition follows from substituting the terms according to \eqref{eq_dist_pi_correct} and \eqref{eq_dist_pi_marginal}.
The above equation asserts that the typical log-ratio $p_x/q_y$ is at most $I$, and the same holds for  $p_y/q_x$. The following simple corollary
essentially follows from Markov's inequality\footnote{One needs to be slightly careful, since the log ratios can in fact be negative, 
while Markov's inequality applies only to non-negative random variables. However, it is well known that the contribution of the negative summands is bounded, 
see \cite{Bra12} for a complete proof.}, so we state it without a proof.

\begin{corollary} \label{lem_ratio_p_q}
Define the set of transcripts $B_\eps := \{ \ell  :  p_x(\ell) > 2^{(I+1)/\eps}\cdot q_y(\ell) \;\; \text{or} \;\;\;  p_y(\ell) > 2^{(I+1)/\eps}\cdot q_x(\ell) \;\; \}$.
Then $\pi_{x,y}(B_\eps) < \eps$.
\end{corollary}
 
The intuitive operational interpretation of the above claim is that, for almost all transcripts $\ell$, the following holds: If a \it uniformly random \rm point $\in[0,1]$ 
falls below $p_y(\ell)$, then the probability it falls below $q_x$ as well is $\gtrsim 2^{-I}$. This intuition gives rise to the following 
rejection sampling approach:
 %Roughy speaking, the compression  scheme we will describe exploits the following operational interpretation of the Kullback-Leiber divergence: 
%``if $\nu,\nu$ are two distributions over the domain $\cU$ and  
%$\Div{\mu}{\nu} = k$, then the probability that a random element $z \in_R \cU$ satisfies
%$\mu(x) $, conditioned that it fell under the  is "   \mnote{rephrase/complete!} \\
The players interpret the public random tape as a sequence of points $(\ell_i,\alpha_i,\beta_i)$, uniformly distributed in $\cU \times [0,1]\times [0,1]$, where $\cU= \{0,1\}^C$
is the set of all possible transcripts of $\pi$. Their goal will be to discover the first %point which ``falls below the histograms of 
index $i^*$ such that $\alpha_{i^*} \leq p_x(\ell_{i^*})$ and $\beta_{i^*} \leq p_y(\ell_{i^*})$. Note that, by design, the probability that a random point $\ell_i$ satisfies these
conditions is precisely $p_x(\ell_{i})\cdot p_y(\ell_{i}) = \pi_{xy}(\ell_i)$, and therefore $\ell_{i^*}$ has the correct distribution. 

The players consider only the first $t := 2|\cU| \ln (1/\eps)$ points of the public tape, as the probability that a single node satisfies the desirable condition is exactly $1/|\cU|$, 
and thus by independence of the points,  the probability that $i^* > t$ is at most $\left(1-1/|\cU|\right)^t = \eps^2 < \eps/16$).

To do so, each player defines his own set of ``potential candidates" for the index $i^*$. Alice defines the set
\[ \cA := \{ i< T \; : \; \alpha_{i} \leq p_x(\ell_{i}) \;\; \text{and} \;\; \beta_{i} \leq 2^{8I/\eps}\cdot q_x(\ell_{i}) \}. \]
Thus $\cA$ is the set of transcript which have the correct distribution on the odd nodes (which Alice can verify by herself), and ``approximately"  satisfies
the desirable condition on the even nodes, on which Alice only has a prior estimate ($q_x$).
Similarly, Bob defines
\[ \cB := \{ i< t \; : \; \beta_{i} \leq p_y(\ell_{i}) \;\; \text{and} \;\; \alpha_{i} \leq 2^{8I/\eps}\cdot q_y(\ell_{i}) \}. \]
By Corollary \ref{lem_ratio_p_q}, $\Pr[\ell^* \notin \cA \cap \cB] \leq \eps/8$, so for the rest of the proof we assume that $\ell^* \in \cA \cap \cB$. In fact, $\ell^*$ is the 
first element of $\cA \cap \cB$. Note that for each point $(\ell_i,\alpha_i,\beta_i)$, $\Pr[\ell_i \in \cA \cap \cB] \leq 2^{8I/\eps}/|\cU|$. Since we consider only the first 
$t = 2|\cU| \ln (1/\eps)$ points, this implies $\E[|\cA|] \leq 2^{8I/\eps}\cdot 2\ln(1/\eps)$, and Chernoff bound further asserts that
\[ \Pr[|\cA|>2^{10I/\eps}] \ll \eps/16. \]
Thus, if we let $\cE_1$ denote the event that $\ell^* \notin \cA \cap \cB$, and $\cE_2 := \{ i^* > t \; \text{or} \; |\cA|>2^{10I/\eps} \; \text{or} |\cB|>2^{10I/\eps} \;  \}$,
then by a union bound $\Pr[\cE_1 \cup \cE_2] \leq  2\eps/8+3\eps/16 < \eps/2$. Thus, letting $\tau_{x,y}$ denote the distribution of $\ell_{i^*} | \neg (\cE_1 \cup \cE_2)$,
the above implies \[ |\tau_{x,y} - \pi_{x,y}| \leq \eps/2, \]
as desired. We will now show a ($2$-round) protocol $\tau$ in which Alice and Bob output a leaf $\ell \sim \tau_{x,y}$, thereby completing the proof.
%because $\tau_{x,y}$ is the distribution $\pi_{x,y}$ restricted to $\cU\setminus \neg(\cA\cup \cB) $ (notice that only the event $\cE_1$ biases the distribution).
To this end, note we have reduced the simulation task to the problem of finding and outputting the first element in $\cA\cap\cB$, where $|\cA|\leq 2^{10I/\eps}$ 
and $|\cB|\leq 2^{10I/\eps}$.  The idea is simple:  Alice wishes to send her entire set $\cA$ to Bob, who can then check for intersection with his set $\cB$. 
Alas, explicitly sending each element $\ell\in \cA$ may be too expensive (requires $\log |\cU|$ bits),
so instead Alice will send Bob sufficiently many ($O(I/\eps)$) random hashes of the elements in $\cA$, using a publicly chosen  sequence of hash functions. 
Since for $a\in \cA$ and $b\in \cB$ such that $a\neq b$, the probability (over the choice of the hash functions) that $h_j(a)=h_j(b)$ for all $j\in O(I/\eps)$ is 
bounded by $2^{-O(I/\eps)} < \frac{\eps}{4|\cA|\cdot|\cB|}$, a union bound ensures that the probability there is an $a\in \cA$, $b\in \cB$ such that $a\neq b$ but the 
hashes happen to match, is bounded by $\eps/4$, which completes the proof. For completeness, the protocol $\tau$ is described in Figure \ref{figure:pi}.

\begin{figure}[h!tb]
\begin{tabular}{|l|}
\hline
\begin{minipage}{\algwidth}
\vspace{1ex}
\begin{center}
\textbf{The simulation protocol $\tau$}
\end{center}
\vspace{0.5ex}
\end{minipage}\\
\hline
\begin{minipage}{\algwidth}
\vspace{1ex}
\begin{enumerate}
    \item Alice computes the set $\cA$. If $|\cA| >2^{10 I/\eps}$ the protocol fails. \item Bob  computes the set $\cB$. If $|\cB| >2^{10 I/\eps}$ the protocol fails. 
    \item For each $a\in \cA$, Alice computes $d=\lceil 20 I/\eps + \log 1/\eps + 2\rceil$ random hash values $h_1(a),\ldots,h_d(a)$, where the hash functions are 
    evaluated using public randomness. 
    \item Alice sends the values $\{h_j(a_i)\}_{a_i\in \cA,~1\le j\le d}$ to Bob. 
    \item Bob finds the first index $i$ such that there is a $b\in \cB$ for which $h_j(b)=h_j(a_i)$ for $j=1..d$ (if such an $i$ exists). Bob outputs $\ell_b$ and sends the index $i$ to Alice. 
    \item Alice outputs $\ell_{i}$. 
\end{enumerate}
\vspace{0.3ex}
\end{minipage}\\
\hline
\end{tabular}
\caption{A simulating protocol for sampling a transcript of $\pi(x,y)$ using $2^{O(I/\eps)}$ communication.}\label{figure:pi}
\end{figure}

\end{proof}

%%%%%%%%%%%%%%%%%%%%%%%%%%%%%%%%%%%%%%%%%%%%%%%%%%%%%%%%%%%%%%%%%%%%%%%%%%%%%%%%%%%%%%%%
%%%%%%%%%%%%%%%%%%%%%%%%%%%%%%%%%%%%%%%%%%%%%%%%%%%%%%%%%%%%%%%%%%%%%%%%%%%%%%%%%%%%%%%%
%%%%%%%%%%%%%%%%%%%%%%%%%%%%%%%%%%%%%%%%%%%%%%%%%%%%%%%%%%%%%%%%%%%%%%%%%%%%%%%%%%%%%%%%
%%%%%%%%%%%%%%%%%%%%%%%%%%%%%%%%%%%%%%%%%%%%%%%%%%%%%%%%%%%%%%%%%%%%%%%%%%%%%%%%%%%%%%%%

\section{Concluding Remarks and Open Problems} \label{sec_openprobs}

We have seen that direct sum and product theorems in communication complexity are essentially equivalent to determining the best possible 
interactive compression scheme. Despite the exciting progress described in this survey, this question is still far from settled, and
the natural open problem is closing the gap in \eqref{eq_ds_current_gap}. The current frontier is trying to improve the dependence 
on $C$ over the scheme of \cite{BBCR}, even at a possible expense of increased dependence on the information cost:
\begin{openproblem}[Improving compression for internal information] \label{op_prob:compression}
Given a protocol $\pi$ over inputs $x,y\sim \mu$, with $\|\pi\|=C, \ICprot{\pi}{\mu} = I$,
is there a communication protocol $\tau$ which  $(0.01)$-simulates $\pi$ such that  $\|\tau \| \leq \poly(I)\cdot C^{1/2 - \eps}$, 
for some absolute positive constant $0<\eps<1/2$?
\end{openproblem}
In fact, by a recent result of Braverman and Weinstein \cite{BW14}, even a much weaker compression scheme in terms of $I$, namely 
$g(I,C) \leq 2^{o(I)}\cdot C^{1/2-\eps}$ would already improve over the the state of the art compression scheme ($\tilde{O}(\sqrt{C\cdot I})$) and 
would imply new direct sum and product theorems. 

Another interesting direction which was unexplored in this survey, is closing the (much smaller) gap in \eqref{eq_bbcr_2}, i.e, 
determining whether a logarithmic dependence on $C$ is essential for interactive compression with
respect to the \emph{external information cost} measure. 
%As argued before, for the direct sum and product implications, 
%one can alternatively view this task as determining the optimal compression 
%rate for internal information under \emph{product distributions}, and therefore this question is essentially 

\begin{openproblem}[Closing the gap for external compression] \label{op_prob:compression}
Given a protocol $\pi$ over inputs $x,y\sim \mu$, with $\|\pi\|=C, \ICprotE{\pi}{\mu} = I$,
is there a communication protocol $\tau$ which  $\delta$-simulates $\pi$ such that  $\|\tau\| \leq \poly(I)\cdot o(\log(C))$?
\end{openproblem}
It is believed that the $(\log C)$ factor is in fact necessary  (see e.g., that candidate separation sampling problem suggested in \cite{Bra13Allerton}), 
but this conjecture remains to be proved.  \\

Recall that in Section \ref{subsec_dp} we saw direct product theorems for randomized communication complexity,
asserting a lower bound on the success rate of  computing $n$ independent copies of $f$ in terms of the success of a single copy. 
When $n$ is very large, such theorems can be superseded by trivial arguments, since $f^n$ must require at least $n$ bits of 
communication just to describe the output. 
%One drawback of such ``hardness amplification"
%technique is that the output of a protocol for $f^n$ is very large ($n$ bits). 
One could hope to achieve hardness amplification without blowing
up the output size -- a classical example is Yao's XOR lemma in circuit complexity. In light of the state-of-the-art direct product result,
we state the following conjecture:
\begin{openproblem}[A XOR Lemma for communication complexity] \label{op_prob:compression}
Is it true that for any $2$-party function $f$ and any distribution $\mu$ on $\cX\times\cY$,
$$ \Dmu{f^{\oplus n}}{\mu^n}{1/2+ e^{-\Omega(n)} } = \tilde{\Omega}(\sqrt{n})\cdot \Dmu{f}{\mu}{2/3} ?$$
(here $f^{\oplus n}((x_1,y_1),\ldots , (x_n,y_n)) := f(x_1,y_1)\oplus ....\oplus f(x_n,y_n)$).
\end{openproblem}

We remark that the ``direct-sum"  analogue of this conjecture is true: \cite{BBCR} proved that their direct sum result for $f^n$
can be easily extended to the computation of  $f^{\oplus n}$,  showing (roughly) that 
$\Dmu{f^{\oplus n}}{\mu^n}{3/4} = \tilde{\Omega}(\sqrt{n})\cdot \Dmu{f}{\mu}{2/3}$. However, this conversion technique does not 
apply to the direct product setting. \\

\section*{Acknowledgements}
I would like to thank Mark Braverman and Oded Regev for helpful discussions and insightful comments on an earlier draft of this survey.
\bibliographystyle{alpha}
\bibliography{refs}

%When *you* (=you) or *I* (=Lane) run latex, the page
%  footers will say something crazy about the issue
%  month/year/volume/number.  But don't worry---when the
%  editor-in-chief runs it, he will have in the right place the right
%  magic file that the style file uses to set these, and so for him it
%  will in theory come out right.

\end{document}